\documentclass[a4paper,11pt,envcountsame]{llncs}  %
 \usepackage[utf8]{inputenc}
\usepackage[english]{babel}
\usepackage[usenames,dvipsnames]{xcolor}
\usepackage{graphicx}
\usepackage{color}
\usepackage{xspace}
\usepackage{placeins}
\usepackage{booktabs}
\usepackage{tabu}
\usepackage{fourier-orns}
\usepackage{subcaption}
\usepackage{amsmath}
\usepackage{amssymb}
\usepackage{mathrsfs}
\usepackage{stmaryrd}
\usepackage{textcomp}
\usepackage{enumerate}
\usepackage[linesnumbered]{algorithm2e}
\usepackage[noend]{algpseudocode}
\usepackage{yfonts}
\usepackage{multirow}
\usepackage{tikz}
\usepackage[symbol]{footmisc}
\usepackage{todonotes}
\usepackage{hyperref}
\usepackage[capitalize,nameinlink,noabbrev]{cleveref}
\usepackage[hmargin=3cm,vmargin=3cm]{geometry} %
\usepackage{bm}

\usetikzlibrary{fit}
\tikzset{%
  hcolor/.style={color=black!20!red},
  highlight/.style={rectangle,rounded corners,draw,dashed,
    thick,inner sep=0pt,hcolor}
}
\newcommand{\tikzmark}[2]{\tikz[overlay,remember picture,
  baseline=(#1.base)] \node (#1) {#2};}
\newcommand{\Highlight}[3]{%
    \tikz[overlay,remember picture]{
    \node[highlight,fit=(#1.north west) (#2.south east)] (submat) {};
    \node[hcolor,below=0mm of submat] {\small#3};
    }
}

\newcommand{\etal}{\mbox{\emph{et al.}}\xspace}
\newcommand\ie{\emph{i.e.}\xspace}
\newcommand\eg{\emph{e.g.}\xspace}
\newcommand\viz{\emph{viz.}\xspace}
\newcommand\defas{:=}

\newcommand\Fq{\ensuremath{\mathbb{F}_q}\xspace}
\newcommand\K{\ensuremath{\mathbb{K}}\xspace}
\newcommand\id{\bm{I}}
\newcommand\zm{\bm{0}}
\newcommand\om{\bm{1}}
\newcommand\diag{\bm{D}}
\newcommand\tpart{\bm{\Delta}}
\newcommand\trig{\bm{T}}

\newcommand\select{\mathit{S}}

\DeclareMathOperator\wt{wt}
\DeclareMathOperator\zr{\mathcal{Z}}

\DeclareMathOperator\bigo{O}
\DeclareMathOperator\gam{\gamma}
\DeclareMathOperator\safe{\mathcal{C}}

\DeclareMathOperator\rowdim{rowdim}
\DeclareMathOperator\coldim{coldim}
\DeclareMathOperator\kbasis{kerb}
\DeclareMathOperator\fun{\mathit{f}}
\DeclareMathOperator\circuit{\mathit{C}}
\DeclareMathOperator\simu{\pi}

\hypersetup{
   colorlinks=true,
   citecolor=ForestGreen,
   linkcolor=BrickRed,
   urlcolor=Fuchsia
 }

\makeatletter\def\@thmcountersep{.}\makeatother
\newcounter{condsec}
\spnewtheorem{condition}{Condition}[condsec]{\bfseries}{\itshape}
\crefname{condition}{Condition}{Conditions}
\newcounter{pcondsec}
\spnewtheorem{precondition}{Precondition}[pcondsec]{\bfseries}{\itshape}
\crefname{precondition}{Precondition}{Preconditions}

\newcommand{\appendixref}[1]{\hyperref[#1]{Appendix~\ref*{#1}}}

\author{Pierre Karpman\inst{1} \and Daniel S. Roche\inst{2}}
\institute{Université Grenoble Alpes, France \and United States Naval Academy, U.S.A.\\
\vspace{2em}
\texttt{pierre.karpman@univ-grenoble-alpes.fr}, \texttt{roche@usna.edu}}

\begin{document}

\title{New Instantiations of the CRYPTO~2017 Masking Schemes}

\maketitle

\begin{abstract}
At CRYPTO~2017, Belaïd \etal presented two new private multiplication algorithms over finite fields, to be used in secure masking schemes.
To date, these algorithms have the lowest known complexity in terms of \emph{bilinear} multiplication and random
masks respectively, both being linear in the number of shares $d+1$. Yet, a practical drawback of both algorithms
is that their safe \emph{instantiation} relies on finding matrices satisfying certain conditions. In their work,
Belaïd \etal only address these up to $d=2$ and 3 for the first and second algorithm respectively, limiting so far the
practical usefulness of their schemes.

In this paper, we use in turn an algebraic, heuristic, and experimental approach to find many more safe instances
of Belaïd \etal's algorithms. This results in explicit such instantiations up to order $d = 6$ over large fields,
and up to $d = 4$ over practically relevant fields such as $\mathbb{F}_{2^8}$. 
\keywords{Masking, linear algebra, MDS matrices.}
\end{abstract}

\section{Introduction}

It has become a well-accepted fact that the black-box security of a cryptographic scheme and the security of one of its real-life implementations may be two quite different matters. In the latter case, numerous
side-channels or fault injection techniques may be used to aid in the cryptanalysis of what could otherwise be a very sensible design (for instance a provably-secure mode of operation on top of a
block cipher with no known dedicated attacks).

A successful line of side-channel attacks is based on the idea of differential power analysis (DPA), which was introduced by Kocher, Jaffe and Jun at CRYPTO'99~\cite{DBLP:conf/crypto/KocherJJ99}.
The practical importance of this threat immediately triggered an effort from cryptographers to find adequate protections. One of the notable resulting counter-measures is
the \emph{masking} approach from Chari \etal and Goubin \& Patarin~\cite{DBLP:conf/crypto/ChariJRR99,DBLP:conf/ches/GoubinP99}. 
The central idea of this counter-measure is to add a ``mask'' to sensitive variables whose observation through a side-channel could otherwise leak secret information;
such variables for instance consist of intermediate values in a block cipher computation that depend on a known plaintext and a round key.
Masking schemes apply a secret-sharing technique to several masked instances of every sensitive variable: a legitimate user knowing
all the shares can easily compute the original value, while an adversary
is now forced to observe more than one value in order to learn anything secret.
The
utility of this overall approach is that it is experimentally the case that the work required to observe $n$ values accurately through DPA increases exponentially
with $n$.

The challenge in masking countermeasures is to find efficient ways to compute with shared masked data while maintaining the property that the observation of
$n$ intermediate values is necessary to learn a secret (for some parameter $n$). When computations are specified as arithmetic circuits over a finite field
\Fq, this task reduces mostly to the specification of secure shared addition and multiplication in that field. A simple and commonly used secret sharing
scheme used in masking is the linear mapping $x \mapsto
\left(r_1,\ldots,r_d,x + \sum_{i=1}^d r_i\right)$ which makes addition
trivial; the problem then becomes how to multiply shared values.
At CRYPTO~2003, Ishai, Sahai and Wagner introduced exactly such a shared multiplication over $\mathbb{F}_2$,
proven secure in a $d$-probing model that they introduced~\cite{DBLP:conf/crypto/IshaiSW03}. Their scheme requires
$d(d+1)/2$ random field elements (\ie bits) and $(d+1)^2$ field multiplications to protect against an adversary able to observe $d$ intermediate values.
This relatively high quadratic complexity in the \emph{order} $d+1$ of the masking lead to an effort to decrease the theoretical and/or practical
cost of masking.

At EUROCRYPT~2016, Belaïd \etal presented a masking scheme over $\mathbb{F}_2$ with \emph{randomness complexity} decreased to $d + d^2/4$; implementations
at low but practically relevant orders $d \leq 4$ confirmed the gain offered by their new algorithm~\cite{DBLP:conf/eurocrypt/BelaidBPPTV16}.
At CRYPTO~2017, the same authors presented two new private
multiplication algorithms over arbitrary finite fields \Fq
\cite{DBLP:conf/crypto/BelaidBPPTV17}.
The first, \emph{Algorithm~4}, decreases the number of \emph{bilinear} multiplications to $2d+1$ at the cost of additional constant multiplications and increased
randomness complexity compared to the algorithm of Ishai \etal; the second, \emph{Algorithm~5}, decreases the randomness complexity to only $d$, at the cost
of $d(d+1)$ constant multiplications. Furthermore, both algorithms are proven secure w.r.t. the strong, composable notions of $d$-(strong) non-interference
from Barthe \etal~\cite{DBLP:conf/ccs/BartheBDFGSZ16}.
Yet a practical drawback of these last two algorithms is that their safe instantiation depends on finding matrices satisfying a certain number of conditions.
Namely, Algorithm~4 uses two (related) matrices in $\Fq^{d\times d}$ for an instantiation at order $d+1$ over \Fq, while Algorithm~5 uses a
single matrix in $\Fq^{d+1 \times d}$ for the same setting. In their paper, Belaïd \etal only succeed in providing ``safe matrices'' for the small cases
$d = 2$ and $d = 2,\,3$ for Algorithms~4 and~5 respectively, and in giving a non-constructive existence theorem for safe matrices when
$q \geq Q = \bigo(d)^{d+1}$ (resp. $q \geq Q = \bigo(d)^{d+2}$).

\subsection{Our contribution}
In this work, we focus on the problem of safely instantiating the two algorithms of Belaïd \etal from CRYPTO~2017.
We first develop equivalent conditions which are in some sense simpler and
much more efficient to check computationally.
We use this reformulation to develop useful
\emph{preconditions} based on MDS matrices that increase the likelihood
that a given matrix is safe. We show
how to generate matrices that satisfy our preconditions by construction,
which then allows to give an explicit sufficient condition, as well as
a construction of safe matrices for both schemes at order $d\le 3$.
Our simplification of the matrix conditions also naturally transforms
into a testing algorithm, an efficient implementation of which is used to perform
an extensive experimental search.
We provide explicit matrices for safe instantiations in all of the
following cases:
\begin{itemize}
  \item For $d=3$, fields $\mathbb{F}_{2^k}$ with $k\ge 3$
  \item For $d=4$, fields $\mathbb{F}_{2^k}$ with $5\le k \le 16$
  \item For $d=5$, fields $\mathbb{F}_{2^k}$ with $10 \le k \le 16$, and additionally $k = 9$ for Algorithm~5. 
  \item For $d=6$, fields $\mathbb{F}_{2^k}$ with $15 \le k \le 16$
\end{itemize}
These are the first known instantiations for $d\ge 4$ or for $d=3$ over
$\mathbb{F}_{2^3}$. We also gather detailed statistics about the proportion of safe matrices in all of these cases.

\subsection{Roadmap}
We recall the two masking schemes of CRYPTO~2017 and the associated matrix conditions in \cref{sec:scheme_def}.
We give our simplifications of the latter in \cref{sec:simp} and state our preconditions in \cref{sec:precond}.
A formal analysis of the case of order up to 3 is given in
\cref{sec:analytic}, where explicit conditions and instantiations for
these orders are also developed. We present
our algorithms and discuss their implementations in \cref{sec:algo}, and conclude with experimental results in \cref{sec:exp}.

\section{Preliminaries}
\label{sec:prelims}

\subsection{Notation}
We use $\K^{m\times n}$ to denote the set of matrices with $m$ rows and
$n$ columns over the field $\K$. We write $m = \rowdim A$ and
$n=\coldim \bm{A}$.
For any vector $\bm{v}$, $\wt(\bm{v})$ denotes the \emph{Hamming
weight} of $\bm{v}$, \ie, the number of non-zero entries.

We use $\zm_{m\times n}$ (resp. $\om_{m\times n}$) to denote the
all-zero (resp. all-one) matrix in $\K^{m\times n}$ for any fixed $\K$
(which will
always be clear from the context). Similarly, $\bm{I}_d$ is the identity matrix of dimension $d$. 

We generally use bold upper-case to denote matrices and bold lower-case
to denote vectors. (The exception is some lower-case Greek letters for
matrices that have been already defined in the literature, notably
$\bm\gamma$.)
For a matrix $\bm{M}$,
$\bm{M}_{i,j}$ is the coefficient at the $i^\text{th}$ row and
$j^\text{th}$ column, with numbering (usually) starting from one.
(Again, $\bm\gamma$ will be an exception as its row numbering starts at 0.)
Similarly, a matrix may be directly defined from its coefficients as
$\begin{pmatrix} \bm{M}_{i,j} \end{pmatrix}$.

We use ``hexadecimal notation'' for binary field elements. This means that $a = \linebreak \sum_{i=0}^{n-1}a_iX^i \in \mathbb{F}_{2^n} \cong \mathbb{F}_2[X]/\langle I(X)\rangle$ (where $I(X)$ is a degree-$n$ irreducible polynomial) is equated to the integer $\tilde{a} = \sum_{i=0}^{n-1}a_i2^i$,
which is then written in base 16. The specific field representations we
use throughout are:
\begin{itemize}
\item $\mathbb{F}_{2^2} \cong \mathbb{F}_2[x]/\langle X^2+X+1\rangle$;
\item $\mathbb{F}_{2^3} \cong \mathbb{F}_2[x]/\langle X^3+X+1\rangle$;
\item $\mathbb{F}_{2^4} \cong \mathbb{F}_2[x]/\langle X^4+X+1\rangle$;
\item $\mathbb{F}_{2^5} \cong \mathbb{F}_2[x]/\langle X^5+X^2+1\rangle$;
\item $\mathbb{F}_{2^6} \cong \mathbb{F}_2[X]/\langle X^6+X+1\rangle$;
\item $\mathbb{F}_{2^7} \cong \mathbb{F}_2[X]/\langle X^7+X+1\rangle$;
\item $\mathbb{F}_{2^8} \cong \mathbb{F}_2[X]/\langle X^8+X^4+X^3+X+1\rangle$;
\item $\mathbb{F}_{2^9} \cong \mathbb{F}_2[X]/\langle X^9+X+1\rangle$;
\item $\mathbb{F}_{2^{10}} \cong \mathbb{F}_2[X]/\langle X^{10}+X^3+1\rangle$;
\item $\mathbb{F}_{2^{11}} \cong \mathbb{F}_2[X]/\langle X^{11}+X^2+1\rangle$;
\item $\mathbb{F}_{2^{12}} \cong \mathbb{F}_2[X]/\langle X^{12}+X^3+1\rangle$;
\item $\mathbb{F}_{2^{13}} \cong \mathbb{F}_2[X]/\langle X^{13}+X^4+X^3+X+1\rangle$;
\item $\mathbb{F}_{2^{14}} \cong \mathbb{F}_2[X]/\langle X^{14}+X^5+1\rangle$;
\item $\mathbb{F}_{2^{15}} \cong \mathbb{F}_2[X]/\langle X^{15}+X+1\rangle$;
\item $\mathbb{F}_{2^{16}} \cong \mathbb{F}_2[X]/\langle X^{16}+X^5+X^3+X+1\rangle$.
\end{itemize}

Additional notation is introduced on first use.

\subsection{MDS \& Cauchy matrices}

An $[n,k,d]_\K$ linear code of length $n$, dimension $k$, minimum distance $d$ over the field $\K$ is \emph{maximum-distance separable} (MDS) if it reaches the Singleton bound,
\ie if $d = n - k + 1$. An \emph{MDS matrix} is the redundancy part $\bm{A}$ of a systematic generating matrix $\bm{G} = \begin{pmatrix} \bm{I}_k & \bm{A}\end{pmatrix}$
of a (linear) MDS code of length double its dimension.

A useful characterization of MDS matrices of particular interest in our case is stated in the following theorem (see \eg~\cite[Chap.~11, Thm.~8]{mdsConj}):
\begin{theorem}
\label{prop:mds_minors}
A matrix is MDS if and only if all its minors are non-zero, \ie all its square sub-matrices are invertible.
\end{theorem}

Square \emph{Cauchy matrices} satisfy the above condition by
construction, and are thence MDS. A (non-necessarily square) matrix
$\bm{A} \in \K^{n\times m}$ is a Cauchy matrix if $\bm{A}_{i,j}
= (x_i - y_j)^{-1}$,
where $\{x_1,\ldots,x_n,y_1,\ldots,y_m\}$ are $n+m$ distinct elements of $\K$.

A Cauchy matrix $\bm{A}$ may be \emph{extended} to a matrix $\widetilde{\bm{A}}$ by adding a row or a column of ones. It can be shown that all square submatrices of $\widetilde{\bm{A}}$ are invertible, and thus
themselves MDS~\cite{DBLP:journals/tit/RothS85}. By analogy and by a slight abuse of terminology, we will say of a square matrix $\bm{A}$ that it is \emph{extended MDS} (XMDS) if all square submatrices of $\bm{A}$
extended by one row or column of ones are MDS. Further depending on the context,
we may only require this property to hold for row (or column) extension
to call a matrix XMDS. 

A (possibly extended) Cauchy matrix $\bm{A}$ may be \emph{generalized} to a matrix $\bm{A'}$ by multiplying it with
(non-zero) row and column scaling: one has $\bm{A'}_{i,j} = c_id_j\cdot(x_i - y_j)^{-1}$, $c_id_j \neq 0$.
All square submatrices of generalized (extended) Cauchy matrices are MDS~\cite{DBLP:journals/tit/RothS85}, but not necessarily XMDS, as one may already use the scaling to set any row or column of $\bm{A'}$ to an arbitrary value.

\subsection{Security notions for masking schemes}

We recall the security notions under which the masking schemes studied in this paper were analysed. These are namely $d$-\emph{non-interference} ($d$-NI) and $d$-\emph{strong non-interference} ($d$-SNI), which were
both introduced by Barthe \etal~\cite{DBLP:conf/ccs/BartheBDFGSZ16} as stronger and composable alternatives to the original $d$-probing model of Ishai \etal~\cite{DBLP:conf/crypto/IshaiSW03}.

Note that none of the notions presented below are explicitly used in this paper, and we only present them for the sake of completeness. Our exposition is strongly based on the one of Belaïd \etal~\cite{DBLP:conf/crypto/BelaidBPPTV17}.

\begin{definition}[Gadgets]
Let $\fun : \mathbb{K}^n \rightarrow \mathbb{K}^m$, $u$, $v\,\in\mathbb{N}$; a \emph{$(u,v)$-gadget} for the function $\fun$ is a randomized circuit $\circuit$ such that for every tuple $(\bm{x}_1,\ldots,\bm{x}_n) \in (\mathbb{K}^u)^n$
and every set of random coins $\mathcal{R}$, $(\bm{y}_1,\ldots,\bm{y}_m) \mapsfrom \circuit(\bm{x}_1,\ldots,\bm{x}_n;\mathcal{R})$ satisfies:
\[
\left(\sum_{j=1}^v \bm{y}_{1,j},\ldots,\sum_{j=1}^v \bm{y}_{m,j}\right) = \fun\left(\sum_{j=1}^u \bm{x}_{1,j},\ldots,\sum_{j=1}^u \bm{x}_{m,j}\right).
\]
One further defines $x_i$ as $\sum_{j=1}^u\bm{x}_{i,j}$, and similarly for $y_i$; $\bm{x}_{i,j}$ is called the $j^\text{th}$ \emph{share} of $x_i$.
\end{definition}

In the above, the randomized circuit $\circuit$ has access to random-scalar gates that generate elements of $\mathbb{K}$ independently and uniformly at random, and the
variable $\mathcal{R}$ records the generated values for a given execution.
Furthermore, one calls \emph{probes} any subset of the wires of $\circuit$ (or equivalently edges of its associated graph).

\begin{definition}[$t$-Simulability]
Let $\circuit$ be a $(u,v)$-gadget for $\fun : \mathbb{K}^n \rightarrow \mathbb{K}^n$, and $\ell$, $t\,\in\mathbb{N}$. A set $\{p_1,\ldots,p_\ell\}$ of probes of $\circuit$ is said to be \emph{$t$-simulable}
if $\exists\, I_1,\ldots,I_n \subseteq \{1,\ldots,u\};\,\#I_i \leq t$ and a randomized function $\simu : (\mathbb{K}^t)^n \rightarrow \mathbb{K}^\ell$ such that for any $(\bm{x}_1,\ldots,\bm{x}_n) \in (\mathbb{K}^u)^n$,
$\{p_1,\ldots,p_\ell\} \sim \{\simu(\{x_{1,i},\,i\in I_1\},\ldots,\{x_{n,i},\,i \in I_n\})\}$.
\end{definition}

This notion of simulability leads to the following.

\begin{definition}[$d$-Non-interference]
A $(u,v)$-gadget $\circuit$ for a function over $\mathbb{K}^n$ is \emph{$d$-non-interfering} (or $d$-NI) if and only if any set of at most $d$ probes of $\circuit$ is $t$-simulable, $t \leq d$. 
\end{definition}

\begin{definition}[$d$-Strong non-interference]
A $(u,v)$-gadget $\circuit$ for a function over $\mathbb{K}^n$ is \emph{$d$-strong non-interfering} (or $d$-SNI) if and only if for every set $P_1$ of at most $d_1$ \emph{internal probes} (that do not depend
on ``output wires'' or output shares $y_{i,j}$'s) and every set $P_2$ of $d_2$ \emph{external probes} (on output wires or shares) such that $d_1 +d_2 \leq d$, then $P_1 \cup P_2$ is $d_1$-simulable.
\end{definition}

It is clear that a $d$-SNI gadget is also $d$-NI. Barthe \etal also showed that the two notions were not equivalent, but that the composition of a $d$-NI and a $d$-SNI gadget was $d$-SNI~\cite{DBLP:conf/ccs/BartheBDFGSZ16}.

\section{The masking schemes of CRYPTO 2017}
\label{sec:scheme_def}

We recall here the main ideas of the two masking schemes of Belaïd \etal introduced at CRYPTO~2017~\cite{DBLP:conf/crypto/BelaidBPPTV17} and their associated matrix conditions;
we refer to that paper for a full description of the gadgets and algorithms.

\subsection{Masking with a linear number of bilinear multiplications~\cite[§4]{DBLP:conf/crypto/BelaidBPPTV17}}
\label{sec:lin_bilin}

This scheme is the composition of two gadgets, only the first of which is of interest to us. In order to build a $d$-SNI multiplication gadget with
$d+1$ input and output shares, Belaïd \etal first give a $d$-NI gadget with $d+1$ input and $2d+1$ output shares, and then compress
its output into $d+1$ shares using a $d$-SNI gadget from Carlet \etal~\cite{DBLP:journals/iacr/CarletPRR16}.

To implement $d$-NI multiplication over a field $\mathbb{K}$, the first gadget
needs a certain matrix $\bm{\gamma} \in \mathbb{K}^{d\times d}$; in turn, this
defines a related matrix ${\bm\delta} \in \mathbb{K}^{d\times d}$ as ${\bm\delta} = \om_{d\times d} - {\bm\gamma}$.
The multiplication algorithm is then derived from the equality:
\begin{multline*}
\label{eqn:lin_bilin}
a \cdot b = \left(a_0 + \sum_{i=1}^d(r_i + a_i)\right) \cdot \left(b_0 + \sum_{i=1}^{d}(s_i + b_i)\right)\\
- \sum_{i=1}^dr_i\cdot\left(b_0+\sum_{j=1}^d({\bm\delta}_{i,j}s_j + b_j)\right) - \sum_{i=1}^d s_i\cdot\left(a_0 + \sum_{j=1}^d({\bm\gamma}_{i,j}r_j+a_j)\right),
\end{multline*}
where $a = \sum_{i=0}^d a_i$, $b = \sum_{i=0}^d b_i$ are the shared multiplicands, and the $r_i$s and $s_i$s are arbitrary (\emph{a priori} random) values.
This equality leads to defining the output shares of this first gadget as:
\begin{itemize}
\item $c_0 \defas \left(a_0 + \sum_{i=1}^d(r_i + a_i)\right) \cdot \left(b_0 + \sum_{i=1}^{d}(s_i + b_i)\right)$;
\item $c_i \defas - r_i\cdot\left(b_0+\sum_{j=1}^d({\bm\delta}_{i,j}s_j + b_j)\right)$, $1 \leq i \leq d$;
\item $c_{i+d} \defas - s_i\cdot\left(b_0+\sum_{j=1}^d({\bm\gamma}_{i,j}s_j + b_j)\right)$, $1 \leq i \leq d$.
\end{itemize}

By considering a proper scheduling of the operations needed to compute the above shares and the probes that this makes available to the adversary,
Belaïd \etal show that a necessary and sufficient condition for their resulting scheme to be $d$-SNI is that ${\bm\gamma}$ and ${\bm\delta}$ \emph{both} satisfy 
a certain condition, stated below.

\setcounter{condsec}{4}\setcounter{condition}{0}
\begin{condition}[\cite{DBLP:conf/crypto/BelaidBPPTV17}]
\label{cond41}
Let ${\bm\gamma} \in \mathbb{K}^{d\times d}$; $\ell = 2d^2+4d + 1$;
$\diag_{{\bm\gamma},j} \in \mathbb{K}^{d\times d}$ be the diagonal matrix whose
non-zero entry at row $i$ is equal to ${\bm\gamma}_{j,i}$; $\trig_d \in
\mathbb{K}^{d\times d}$
be the upper-triangular matrix whose non-zero entries are all one; and
$\trig_{{\bm\gamma},j} \in \mathbb{K}^{d\times d} = \diag_{{\bm\gamma},j}\trig_d$. Equivalently:
\begin{align*}
\id_d = \begin{pmatrix}
1 & 0 & \cdots & 0\\
0 & 1 &  & 0\\
\vdots &  & \ddots & \vdots\\
0 & \cdots & 0 & 1\\
\end{pmatrix}, &\qquad
\diag_{{\bm\gamma},j} = \begin{pmatrix}
{\bm\gamma}_{j,1} & 0 & \cdots & 0\\
0 & {\bm\gamma}_{j,2} &  & 0\\
\vdots &  & \ddots & \vdots\\
0 & \cdots & 0 & {\bm\gamma}_{j,d}\\
\end{pmatrix},\\
\\
\trig_d = \begin{pmatrix}
1 & 1 & \cdots & 1\\
0 & 1 & \cdots & 1\\
\vdots &  & \ddots & \vdots\\
0 & \cdots & 0 & 1\\
\end{pmatrix}, &\qquad
\trig_{{\bm\gamma},j} = \begin{pmatrix}
{\bm\gamma}_{j,1} & {\bm\gamma}_{j,1} & \cdots & {\bm\gamma}_{j,1}\\
0 & {\bm\gamma}_{j,2} & \cdots & {\bm\gamma}_{j,2}\\
\vdots &   & \ddots & \vdots\\
0 & \cdots & 0 & {\bm\gamma}_{j,d}\\
\end{pmatrix}.\\
\end{align*}
One then defines $\bm{L} \in \mathbb{K}^{(d+1)\times\ell}$ and
$\bm{M_{\bm\gamma}} \in \mathbb{K}^{d\times\ell}$ as:
\[\begin{matrix}
  \bm{L} &=& \bigg( &
    \begin{matrix}
      1 \\ \zm_{d\times 1}
    \end{matrix}&\begin{matrix}
      \zm_{1\times d} \\ \id_d
    \end{matrix}&\begin{matrix}
      \zm_{1\times d} \\ \zm_{d\times d}
    \end{matrix}&\begin{matrix}
      \zm_{1\times d} \\ \id_d
    \end{matrix}&\begin{matrix}
      \zm_{1\times d} \\ \id_d
    \end{matrix}&\cdots&\begin{matrix}
      \zm_{1\times d} \\ \id_d
    \end{matrix}&\begin{matrix}
      \om_{1\times d} \\ \trig_d
    \end{matrix}&\begin{matrix}
      \om_{1\times d} \\ \trig_d
    \end{matrix}&\cdots&\begin{matrix}
      \om_{1\times d} \\ \trig_d
    \end{matrix}&
  \bigg), \\[6 mm]
  \bm{M}_{\bm\gamma} &=& ( &
    \zm_{d\times 1} &
    \zm_{d\times d} &
    \id_d &
    \id_d &
    \diag_{{\bm\gamma},1} &
    \cdots &
    \diag_{{\bm\gamma},d} &
    \trig_d &
    \trig_{{\bm\gamma},1} &
    \cdots &
    \trig_{{\bm\gamma},d} &
  ).
\end{matrix}\]
Finally, ${\bm\gamma}$ is said to satisfy \cref{cond41} if for any
vector $\bm{v} \in \mathbb{K}^\ell$ of Hamming weight $\wt(\bm{v}) \leq
d$ such that $\bm{L}\bm{v}$ contains
no zero coefficient (\ie is of maximum Hamming weight $d+1$), then $\bm{M}_{\bm\gamma}\bm{v}\neq\zm_{d\times 1}$.
\end{condition}
An equivalent, somewhat more convenient formulation of \cref{cond41} can be obtained by contraposition: ${\bm\gamma}$ satisfies \cref{cond41} if
\begin{equation}
\label{eqn:equiv41}
\bm{v} \in \ker(\bm{M}_{\bm\gamma}) \wedge \wt(\bm{v}) \leq d
\Rightarrow \wt(\bm{L}\bm{v}) < d + 1.
\end{equation}
Whichever formulation is adopted, the logic behind this condition is that a violation of the implication means that there exists a linear combination of at most $d$ probes 
that depends on all the input shares (as $\bm{L}\bm{v}$ is of full weight) and on no random mask (as $\bm{M}_{\bm\gamma}\bm{v} = \zm_{d\times1}$).
In that respect, $\bm{L}$ and $\bm{M}$ behave as ``indicator matrices'' for the shares and masks on which depend individual probes.

\subsection{Masking with linear randomness complexity~\cite[§5]{DBLP:conf/crypto/BelaidBPPTV17}}
\label{sec:lin_rand}
The second scheme that we consider is defined by a single $d$-NI
multiplication gadget over $\mathbb{K}$ that has $(d+1)$ input and output shares.
An instantiation depends on a matrix
${\bm\gamma} \in \mathbb{K}^{(d+1)\times d}$ whose rows sum to zero, \ie, such that $\sum_{i=0}^{d} {\bm\gamma}_i = \zm_{1\times d}$.\footnote{Note that for convenience in the subsequent share definitions and consistency
with the notation of~\cite{DBLP:conf/crypto/BelaidBPPTV17}, the row
index of ${\bm\gamma}$ starts from zero and not one.} This lets us defining the output shares as:
\begin{itemize}
\item $c_i = a_0b_i + \sum_{j = 1}^d({\bm\gamma}_{i,j}r_j + a_jb_i)$, $0 \leq i \leq d$,
\end{itemize}
where again $a = \sum_{i=0}^d a_i$, $b = \sum_{i=0}^d b_i$ are the shared multiplicands and the $r_i$s are arbitrary values.

Belaïd \etal show that a necessary and sufficient condition for their resulting gadget to be $d$-NI is that ${\bm\gamma}$ satisfies a condition similar to \cref{cond41}, stated below.

\setcounter{condsec}{5}\setcounter{condition}{0}
\begin{condition}[\cite{DBLP:conf/crypto/BelaidBPPTV17}]\label{cond51}
Let ${\bm\gamma} \in \mathbb{K}^{(d+1)\times d}$ $\ell$, $\diag_{{\bm\gamma},j}$, $\trig_d$, $\trig_{{\bm\gamma},j}$ be as in \cref{cond41} and
$\mathbb{K}(\omega_0,\ldots,\omega_d)$ be the field of rational fractions over indeterminates $\omega_0,\ldots,\omega_d$.
One defines $\bm{L'} \in
\mathbb{K}(\omega_0,\ldots,\omega_d)^{(d+1)\times\ell}$ and $\bm{M}'_{\bm\gamma}
\in \mathbb{K}^{d\times\ell}$ as:
\[\begin{matrix}
  \bm{L'} &=& \bigg( &
    \begin{matrix}
      1 \\ \zm_{d\times 1}
    \end{matrix}&\begin{matrix}
      \zm_{1\times d} \\ \id_d
    \end{matrix}&\begin{matrix}
      \zm_{1\times d} \\ \zm_{d\times d}
    \end{matrix}&\begin{matrix}
      \zm_{1\times d} \\ \omega_0\id_d
    \end{matrix}&\begin{matrix}
      \zm_{1\times d} \\ \omega_1\id_d
    \end{matrix}&\cdots&\begin{matrix}
      \zm_{1\times d} \\ \omega_d\id_d
    \end{matrix}&\begin{matrix}
      \omega_0\om_{1\times d} \\ \omega_0\trig_d
    \end{matrix}&\begin{matrix}
      \omega_1\om_{1\times d} \\ \omega_1\trig_d
    \end{matrix}&\cdots&\begin{matrix}
      \omega_d\om_{1\times d} \\ \omega_d\trig_d
    \end{matrix}&
  \bigg), \\[6 mm]
  \bm{M'}_{\bm\gamma} &=& ( &
    \zm_{d\times 1} &
    \zm_{d\times d} &
    \id_d &
    \diag_{{\bm\gamma},0} &
    \diag_{{\bm\gamma},1} &
    \cdots &
    \diag_{{\bm\gamma},d} &
    \trig_{{\bm\gamma},0} &
    \trig_{{\bm\gamma},1} &
    \cdots &
    \trig_{{\bm\gamma},d} &
  ).
\end{matrix}\]
Then ${\bm\gamma}$ is said to satisfy \cref{cond51}
if for any vector $\bm{v} \in \mathbb{K}^\ell$ of Hamming weight
$\wt(\bm{v}) \leq d$ such that $\bm{L'}\bm{v}$ contains
no zero coefficient, then $\bm{M'}_{\bm\gamma}\bm{v}\neq\zm_{d\times 1}$.
\end{condition}

Note that as $\mathbb{K}$ is a subfield of $\mathbb{K}(\omega_0,\ldots,\omega_d)$
(\viz the field of its constants), the product $\bm{L'}\bm{v}$ is well-defined.
Also, again by contraposition, \cref{cond51} can be expressed as:

\begin{equation}
\label{eqn:equiv51}
\bm{v} \in \ker(\bm{M'}_{\bm\gamma}) \wedge \wt(\bm{v}) \leq d
\Rightarrow \wt(\bm{L'}\bm{v}) < d + 1.
\end{equation}

\section{Simplifying and unifying the conditions}
\label{sec:simp}

In this section, we describe a few simplifications and consolidations of the
correctness and safety for the two schemes described in the previous section.
These simplifications are important for our analytical and algorithmic
results, and the consolidations of the two schemes allow for ease in
presentation.

Specifically, we develop three related conditions $\safe$, $\safe'$, and
$\safe''$, on the matrices
$\bm{M}_{\bm\gamma}$, $\bm{L}_d$, $\bm{M}'_{\bm\gamma}$, and
$\bm{L}'_d$ defined in \cref{cond41,cond51}, such that the safety of
the masking schemes is guaranteed when these conditions are true.
We prove that the first condition $\safe$ and the third condition
$\safe''$ are both exactly equivalent to the requirements of
\cref{cond41,cond51}. The second condition $\safe'$ is always a
\emph{sufficient} condition as it implies the other two, and it is
also \emph{necessary} under a very mild condition on the field size.

\subsection{Unifying $\bm{M}_{\bm\gamma}$ and $\bm{M'}_{\bm\gamma}$}

Recall the definitions of matrices $\bm{M}_{\bm\gamma}$ from
\cref{cond41} and $\bm{M'}_{\bm\gamma}$ from \cref{cond51}. These
are both $d\times \ell$ matrices (where $\ell = 2d^2+4d+1$) consisting
of zeros, ones, and entries from ${\bm\gamma}$. Moreover, $\bm{M}_{\bm\gamma}$
and $\bm{M'}_{\bm\gamma}$ are exactly the same except for in one
submatrix of $d$ columns: this submatrix is $\bm{T}_d$ in
$\bm{M}_{\bm\gamma}$ and $\bm{T}_{{\bm\gamma},0}$ in $\bm{M'}_{\bm\gamma}$.

We can unify these two matrices by considering, in the case of
\cref{cond41}, augmenting the ${\bm\gamma}$ matrix with an additional row of
1's at index 0. Then $\bm{T}_d = \bm{T}_{{\bm\gamma},0}$ and we can
consider only the second form of the matrix $\bm{M'}_{\bm\gamma}$.

Note that the corresponding matrices $\bm{L_{\bm\gamma}}$ and
$\bm{L'_{\bm\gamma}}$ from \cref{cond41,cond51} respectively are still
not identical, but the locations of non-zero entries (i.e., the
\emph{support}) in $\bm{L_{\bm\gamma}}$ and $\bm{L'_{\bm\gamma}}$ are the
same.

Now for both schemes, there is a single matrix
${\bm\gamma}\in\mathbb{K}^{(d+1)\times d}$ which determines their
\emph{correctness} (do the output shares always
correspond to the multiplication of the input value) and
\emph{safety} (is it possible for an attacker to learn any secrets with
at most $d$ probes).

To succinctly state the unified condition, we first define a simple
predicate $\zr$ for when a matrix $\bm{X}\in\K^{m\times
n}$ (or column vector $\bm{x}\in\K^m$)
has at least one row of zeros:
\[
\zr(\bm{X}) := \exists\,i \in \{1,\ldots,m\} \text{ s.t. }
\forall\,j \in \{1,\ldots,n\}, \bm{X}_{i,j} = 0.
\]

Based on the above discussion, we define the
following crucial predicate for the safety definition for two arbitrary
matrices $\bm{A}$ and $\bm{B}$ with the same number of columns:
\begin{equation}\label{eqn:safe}
\safe(\bm{A},\bm{B}) :=
  \forall\, \bm{v} \in \ker(\bm{A}) \text{ s.t. }
  \wt(\bm{v}) \le \rowdim(\bm{A}),\text{ then }
  \zr(\bm{B}\bm{v}).
\end{equation}
Typically we will have $\bm{A}=\bm{M}'_{\bm\gamma}$ and $\bm{B}$
is either $\bm{L}$ or $\bm{L'}$.

Now we can restate the correctness and safety conditions for the two
schemes. The following propositions follow directly from the definitions and
discussions so far.

\begin{proposition}\label{prop:scheme4}
  For ${\bm\gamma}\in\mathbb{K}^{(d+1)\times d}$, the scheme of
  \cref{sec:lin_bilin} is correct and safe if and only if
  the following conditions are met, where
  $\displaystyle
    {\bm\delta} = \begin{pmatrix}\bm{2}_{1\times d}\footnotemark \\ \om_{d\times d}\end{pmatrix} -
    {\bm\gamma}
  $\footnotetext{In fields of characteristic 2, the matrix $\bm{2}_{1\times d}$ is
  actually $\zm_{1\times d}$.}:
  \begin{enumerate}[(1)]
    \item ${\bm\gamma}_{0,j} = 1\text{ for all }j \in \{1,\ldots,d\}$
    \item $\safe(\bm{M'}_{\bm\gamma},\bm{L})$
    \item $\safe(\bm{M'}_{\bm\delta},\bm{L})$
  \end{enumerate}
\end{proposition}

\begin{proposition}\label{prop:scheme5}
  For ${\bm\gamma}\in\mathbb{K}^{(d+1)\times d}$, the scheme of
  \cref{sec:lin_rand} is correct and safe if and only if the
  following conditions are met:
  \begin{enumerate}[(1)]
    \item $\sum_{i=0}^{d} {\bm\gamma}_{i} = \bm{0}_{1\times d}$ \\
    \item $\safe(\bm{M'}_{\bm\gamma},\bm{L'})$
  \end{enumerate}
\end{proposition}

\subsection{Equivalent condition with kernel bases}

Next we develop a condition similar to the definition of
$\safe(\bm{A},\bm{B})$
as defined in \eqref{eqn:safe} above, but in terms of kernel bases rather
than individual vectors. This modified condition is equivalent under a
mild requirement on the size of the field $\K$.

The general idea is that rather than considering all matrix-vector
products
$\bm{B}\bm{v}$, where $\bm{v}$ is a $d$-sparse vector in the
right kernel of $\bm{A}$,
we consider instead the kernel basis for a size-$d$ subset of $\bm{A}$'s
columns, and multiply the corresponding columns in $\bm{B}$ times
this basis.
Specifying this condition requires some additional notation which will
also be useful later on.

Let $\kbasis(\bm{X})$ denote a basis of the right kernel of
$\bm{X}$. That is, any vector $\bm{v}\in\ker(\bm{X})$ is a
linear combination of the columns of $\kbasis(\bm{X})$.

Let $[c_1,\ldots,c_k]$ be a list of $k$ distinct column
indices, where each $1\le c_i \le \ell$. Selecting only these columns
from any matrix with $\ell$ columns is a linear operator corresponding
to a \emph{selection matrix} $\bm{P}\in\{0,1\}^{\ell\times k}$, where
$\bm{P}_{i,j}=1$ iff $c_j=i$. Define $\select^\ell_m$ as
the set of all $\ell\times m$ selection matrices. That is,
$\select^\ell_m$ consists of all $\{0,1\}$-matrices with $\ell$ rows
and at most $m$ columns, where there is a single $1$ in each column and
no two $1$s in the same row.

Note that the product of a selection matrix and its
transpose is an identity matrix with some rows and columns set to zero.
For any matrix (or vector) $\bm{X}\in\K^{m\times n}$
with at most $k$ non-zero rows,
there is a selection matrix
$\bm{P}\in\select^{k}_m$ such that
$\bm{P}\bm{P}^T\bm{X} = \bm{X}$.

The equivalent condition to \eqref{eqn:safe} that we consider now is
formed by multiplying some subset of $\bm{B}$'s columns times a
kernel basis of the same subset of $\bm{A}$'s columns:
\begin{equation}\label{eqn:safeprime}
  \safe'(\bm{A},\bm{B})\ := \ 
    \forall\, \bm{P} \in \select^\ell_{\rowdim(\bm{A})},\ 
    \zr\mathopen{}\left(\bm{B}\bm{P}\cdot
    \kbasis\mathopen{}\left(\bm{A}\bm{P}\right)\mathclose{}
    \right)\mathclose{}.
\end{equation}

One direction of the
equivalence is straightforward, and the other depends on the
Schwartz-Zippel lemma and therefore on the size of the field.
Even so, the field size requirement here is very mild; indeed the field
is sufficiently large in all cases where we are aware of any valid
constructions of the schemes.

\begin{theorem}\label{thm:safeprime}
  For any $\bm{A}\in\K^{n\times\ell}$ and
  $\bm{B}\in\K^{m\times\ell}$, we have
  $\safe'(\bm{A},\bm{B})\Rightarrow\safe(\bm{A},\bm{B})$.
  If $\K$ has at least $m+1$ distinct elements, then
  $\safe'(\bm{A},\bm{B})\Leftarrow\safe(\bm{A},\bm{B})$
  also.
\end{theorem}
\begin{proof}
  We begin with the ``$\Rightarrow$'' direction.

  Let $\bm{v}$ be a vector satisfying the conditions of
  $\safe(\bm{A},\bm{B})$; that is, $\bm{v}\in\ker\bm{A}$
  and $\wt(\bm{v})\le\rowdim(\bm{A})$. The latter fact means that
  there exists $\bm{P}\in\select^\ell_{\rowdim(\bm{A})}$
  such that
  $\bm{P}\bm{P}^T\bm{v}=\bm{v}$.

  Because $\bm{A}\bm{v}=\zm$, we then have
  $(\bm{A}\bm{P})(\bm{P}^T\bm{v)}=\zm$, which means
  that the vector $\bm{P}^T\bm{v}$ is a linear combination of
  the columns of $\kbasis(\bm{A}\bm{P})$.

  The condition $\safe(\bm{A},\bm{B})$ concerns the
  matrix-vector product $\bm{B}\bm{v}$, which equals
  $\bm{B}\bm{P}\bm{P}^T\bm{v}$. From above, we know that
  this is a linear combination of the columns in the matrix
  $\bm{B}\bm{P}\cdot\kbasis(\bm{A}\bm{P})$. By the
  assumption that $\safe'(\bm{A},\bm{B})$, this matrix contains a
  zero row, and therefore any linear combination of its columns also
  contains a zero row; hence $\zr(\bm{B}\bm{v})$.

  For the ``$\Leftarrow$'' direction, we prove using the contrapositive.
  Assume there exists some selection of columns
  $\bm{P}\in\select^\ell_n$ such that
  $\neg\zr(\bm{B}\bm{P}\cdot\kbasis(\bm{A}\bm{P}))$.
  We need to show
  that $\neg\safe(\bm{A},\bm{B})$.

  Suppose the column dimension of $\kbasis(\bm{A}\bm{P})$ (\ie, the
  nullity of $\bm{A}\bm{P}$) is $k$, and let $\bm{x}$ be
  a column vector of $k$ indeterminates $x_1,\ldots,x_k$.
  Now consider the matrix-vector
  product
  $\bm{B}\bm{P}\cdot\kbasis(\bm{A}\bm{P})\cdot\bm{x}$.
  This is a
  column vector of dimension $m$ consisting of degree-$1$ polynomials
  in the $k$ indeterminates. Furthermore, none of these polynomials is
  zero because of the assumption
  $\neg\zr(\bm{B}\bm{P}\cdot\kbasis(\bm{A}\bm{P}))$.

  The product of the $m$ polynomials in
  $\bm{B}\bm{P}\cdot\kbasis(\bm{A}\bm{P})\cdot\bm{x}$
  is a single
  non-zero polynomial in $k$ variables with total degree $m$.
  By the Schwartz-Zippel-DeMillo-Lipton lemma
  \cite[Cor.~1]{DBLP:journals/jacm/Schwartz80},
  and because $\#\K > m$, there must exist some assignment of the $k$
  variables to values in $\K$ such that this product polynomial is non-zero.
  That is, there exists some column vector $\bm{w}\in\K^k$ such that
  $
  \wt(\bm{B}\bm{P}\cdot\kbasis(\bm{A}\bm{P})\cdot\bm{w}) = m.
  $

  Because $\kbasis(\bm{A}\bm{P})\cdot\bm{w}\in\K^n$,
  there is an $n$-sparse vector $\bm{v}\in\K^\ell$
  such that
  $\bm{P}^T\bm{v}=\kbasis(\bm{A}\bm{P})\cdot\bm{w}$.
  This vector $\bm{v}$ shows that
  $\safe(\bm{A},\bm{B})$ is false. Namely,
  $\bm{v} \in \ker(\bm{A})$ because
  $\bm{A}\bm{v}=(\bm{A}\bm{P})(\bm{P}^T\bm{v})=\zm$;
  it has low weight $\wt(\bm{v}) \le n$; and
  $\bm{B}\bm{v}=(\bm{B}\bm{P})(\bm{P}^T\bm{v})$ is
  of full weight $m$ from the previous paragraph.
  \qed
\end{proof}

\subsection{Eliminating rows and columns}

The third simplification to the correctness and safety conditions of the
two masking schemes that we develop is an equivalent condition to
$\safe(\bm{M}'_{\bm\gamma},\bm{N})$
that depends on less than half of the columns in the original
matrices. The intuition is that most of the columns of these matrices
have weight 1, and thus those probes in the masking scheme do not
gain the attacker any real advantage. So we can focus on only
the parts of $\bm{M}'_{\bm\gamma}$ and $\bm{N}$ whose
columns have weight greater than 1. We first develop some new
terminology to talk about these submatrices, then prove a lemma which
shows how to eliminate columns from ${\bm\gamma}$ corresponding to the
weight-one probes, and finally state and prove the equivalent condition
$\safe''$.

So far the schemes are both defined by a matrix ${\bm\gamma}$ with $d+1$ rows
and $d$ columns. In fact, the definitions of matrices $\bm{M}_{\bm\gamma}$,
$\bm{M'}_{\bm\gamma}$, $\bm{L}$, and $\bm{L'}$ from
\cref{cond41,cond51} generalize to any rectangular matrix
${\bm\gamma}\in\K^{(d+1)\times n}$.
If ${\bm\gamma}$ has $d+1$ rows and $n$ columns,
then $\bm{M}_{\bm\gamma}$ and $\bm{M}'_{\bm\gamma}$ both have $n$ rows,
while $\bm{L}_n$ and $\bm{L}'_n$ have $n+1$ rows, and all four
matrices have $\ell_n=2dn+4n+1$ columns.

We focus on the bottom-right
$n\times(dn+n)$ submatrix of each $\bm{M}'_{\bm\gamma}$, $\bm{L}_n$
and $\bm{L}'_n$, which we call the
``triangular part'' of each. Formally, we define a linear operator
$\tpart$ such that, for any matrix $\bm{A}$
with $n$ or $n+1$ rows and $2nd+4d+1$ columns,
$\tpart(\bm{A})$ consists of the bottom-right $n\times(dn+n)$
submatrix of $\bm{A}$.

In summary, we have:

\[\begin{matrix}
  \bm{L}_n &=& \bigg( &
    \begin{matrix}
      1 \\ \zm_{n\times 1}
    \end{matrix}&\begin{matrix}
      \zm_{1\times n} \\ \id_n
    \end{matrix}&\begin{matrix}
      \zm_{1\times n} \\ \zm_{n\times n}
    \end{matrix}&\begin{matrix}
      \zm_{1\times n} \\ \id_n
    \end{matrix}&\begin{matrix}
      \zm_{1\times n} \\ \id_n
    \end{matrix}&\cdots&\begin{matrix}
      \zm_{1\times n} \\ \id_n
    \end{matrix}&\begin{matrix}
      \om_{1\times n} \\ \tikzmark{leftL}{$\trig_n$}
    \end{matrix}&\begin{matrix}
      \om_{1\times n} \\ \trig_n
    \end{matrix}&\cdots&\begin{matrix}
      \om_{1\times n} \\ \tikzmark{rightL}{$\trig_n$}
    \end{matrix}&
  \bigg),\\[28pt]
  \bm{L'}_n &=& \bigg( &
    \begin{matrix}
      1 \\ \zm_{n\times 1}
    \end{matrix}&\begin{matrix}
      \zm_{1\times n} \\ \id_n
    \end{matrix}&\begin{matrix}
      \zm_{1\times n} \\ \zm_{n\times n}
    \end{matrix}&\begin{matrix}
      \zm_{1\times n} \\ \omega_0\id_n
    \end{matrix}&\begin{matrix}
      \zm_{1\times n} \\ \omega_1\id_n
    \end{matrix}&\cdots&\begin{matrix}
      \zm_{1\times n} \\ \omega_d\id_n
    \end{matrix}&\begin{matrix}
      \omega_0\om_{1\times n} \\ \tikzmark{leftLp}{$\omega_0\trig_n$}
    \end{matrix}&\begin{matrix}
      \omega_1\om_{1\times n} \\ \omega_1\trig_n
    \end{matrix}&\cdots&\begin{matrix}
      \omega_d\om_{1\times n} \\ \tikzmark{rightLp}{$\omega_d\trig_n$}
    \end{matrix}&
  \bigg),\\[34pt]
  \bm{M'}_{\bm\gamma} &=& ( &
    \zm_{n\times 1} &
    \zm_{n\times n} &
    \id_n &
    \diag_{{\bm\gamma},0} &
    \diag_{{\bm\gamma},1} &
    \cdots &
    \diag_{{\bm\gamma},d} &
    \tikzmark{leftMp}{$\trig_{{\bm\gamma},0}$} &
    \trig_{{\bm\gamma},1} &
    \cdots &
    \tikzmark{rightMp}{$\trig_{{\bm\gamma},d}$} &
  ).\\[15pt]
\end{matrix}
\Highlight{leftL}{rightL}{$\tpart(\bm{L}_n)$}
\Highlight{leftLp}{rightLp}{$\tpart(\bm{L}'_n)$}
\Highlight{leftMp}{rightMp}{$\tpart(\bm{M}'_{\bm\gamma})$}
\]

Notice that the matrices $\bm{L}_n$ and $\bm{L}'_n$ have some
different entries but the same non-zero locations or \emph{support};
for convenience
we denote by $\bm{N}_n$ any matrix with this same dimension and
support.

Inspecting the definition of $\bm{M'}_{\bm\gamma}$, we see that rows of
this matrix correspond to columns of ${\bm\gamma}$, and removing the $i$th
column of ${\bm\gamma}$ corresponds to removing a single row and $2d+4$
columns from each of
$\bm{M'}_{\bm\gamma}$ and $\bm{N}$.

Notice also that the columns of $\bm{M}'_{\bm\gamma}$ and of
$\bm{L}_n$ which are not in the triangular parts all have weight at
most one. This means, as we show in the following technical lemma, that
the effect of any such column choice can be eliminated by
removing one row each from $\bm{M}'_{\bm\gamma}$ and $\bm{L}_n$. In
terms of masking schemes, this means that a single probe corresponding
to these non-triangular parts allows the adversary to cancel at most one
random value and to learn at most one share. Because the number of
shares is $d+1$ in a scheme allowing $d$ probes, this results in no
advantage for the adversary.

\begin{lemma}\label{lem:removecol}
  Let ${\bm\gamma}\in\K^{(d+1)\times n}$, $\bm{M}'_{\bm\gamma}$ and
  $\bm{N}_n$ be as above.
  Suppose $\bm{u}\in\K^{\ell_n}$ is a vector with
  $\wt(\bm{u})=1$ whose single non-zero entry is between index $2$ and
  $dn+3n+1$ inclusive, and $\bm{v}\in\K^{\ell_n}$ is any other
  vector.
  Then there exists a selection matrix $\bm{P}\in\select^n_{n-1}$
  and another vector $\bm{w}\in\K^{\ell_{n-1}}$ with
  $\wt(\bm{w})\le\wt(\bm{v})$ such that
  $$
    \wt(\bm{M}'_{{\bm\gamma}\bm{P}}\bm{w}) \le
      \wt(\bm{M}'_{{\bm\gamma}}(\bm{u}+\bm{v}))
    \quad\text{and}\quad
    \wt(\bm{N}_{n-1}\bm{w}) \ge
      \wt(\bm{N}_n(\bm{u}+\bm{v})) - 1.
  $$
\end{lemma}
\begin{proof}
  Write $i$ for the index of the non-zero entry in $\bm{u}$. We can
  see that the $i$th column of $\bm{M}'_{\bm\gamma}$ and $\bm{N}_n$
  both have weight at most one. Indeed, for each
  $i\in\{2,\ldots,dn+3n+1\}$, there is a corresponding index
  $j\in\{1,\ldots,n\}$ such that the $i$th columns of
  $\bm{M}'_{\bm\gamma}$ and $\bm{N}_n$ are zero everywhere except
  possibly in row $j$ (provided that we continue to index the rows of
  $\bm{N}_n$ starting at 0).

  Removing the $j$th row from $\bm{M}'_{\bm\gamma}$ and $\bm{N}_n$
  results in two new matrices $\bm{A},\bm{B}$ (respectively)
  whose $i$th columns are both zero, and hence
  $\bm{A}\bm{u}=\zm$ and $\bm{B}\bm{u}=\zm$.
  This means that
  \begin{align*}
    \wt(\bm{A}\bm{v}) &= \wt(\bm{A}(\bm{u}+\bm{v}))
      \le \wt(\bm{M}'_{\bm\gamma}(\bm{u}+\bm{v})) \\
    \wt(\bm{B}\bm{v}) &= \wt(\bm{B}(\bm{u}+\bm{v}))
      \ge \wt(\bm{N}_n(\bm{u}+\bm{v})) - 1.
  \end{align*}

  Write $\bm{P}\in\select^n_{n-1}$ as the matrix which selects
  all $n$ columns of ${\bm\gamma}$ except for the $j$th column.
  Now $\bm{A}$ and $\bm{B}$ are the same as
  $\bm{M}'_{{\bm\gamma}\bm{P}}$ and $\bm{N}_{n-1}$ respectively,
  except that they each have $2d+4$ extra columns.
  The remaining task is to modify $\bm{v}$ so that it is zero at all
  the indices corresponding to these extra columns, without changing
  $\wt(\bm{A}\bm{v})$ or $\wt(\bm{B}\bm{v})$.

  We can see that $d+3$ of these extra columns come from the first
  $dn+3n+1$ columns of $\bm{M}'_{\bm\gamma}$ and $\bm{N}_n$ and,
  since the $j$th row has been removed, they are in fact now zero
  columns. So letting $\bm{v}'$ be the same as $\bm{v}$ with any
  such entries set to zero, we do not change the products
  $\bm{A}\bm{v}'$ or $\bm{B}\bm{v}'$ at all.

  The $d+1$ remaining extra columns come from the triangular parts
  $\tpart(\bm{M}'_{\bm\gamma})$ and $\tpart(\bm{N}_n)$. There are now
  two cases to consider. First, if $j=1$, \ie, we have removed the
  second row of $\bm{N}_n$ and the first row of
  $\bm{M}'_{\bm\gamma}$. Then these extra columns from the triangular
  part of $\bm{A}$ are all zero columns, and from $\bm{B}$ they
  have the form $(a\ 0\ \cdots\ 0)^T$
  for some non-zero entry $a$ in the first row of $\bm{N}_n$.
  Upon inspection, we see that these columns are exactly $a$ times the
  very first columns of $\bm{A}$ and $\bm{B}$ respectively.
  Therefore we can modify the vector $\bm{v}'$ to a new vector
  $\bm{v}''$, where any non-zero entries in such positions are
  divided by $a$ and added to the first entry, then set to zero. This
  does not change the value of
  $\bm{A}\bm{v}''$ or $\bm{B}\bm{v}''$.

  The second case is that $j\ge 2$, \ie, we have removed a later row.
  Then the extra columns in $\bm{A}$ and $\bm{B}$ are exactly
  identical to the columns immediately to their left in the respective
  matrices. So we can form $\bm{v}''$ in this case by adding
  any non-zero entry of $\bm{v}'$ in such positions to the adjacent
  position and then setting it to zero, without changing
  $\bm{A}\bm{v}''$ or $\bm{B}\bm{v}''$.

  After this, we have a vector $\bm{v}''$ with $\wt(\bm{v}'')
  \le \wt(\bm{v})$, and with zeros in all of the ``extra column''
  indices of $\bm{A}$ and $\bm{B}$, such that
  $\wt(\bm{A}\bm{v}'') \le
    \wt(\bm{M}'_{\bm\gamma}(\bm{u}+\bm{v}))$ and
  $\wt(\bm{B}\bm{v}'') \ge
    \wt(\bm{N}_n(\bm{u}+\bm{v}))-1$.
  Finally, setting $\bm{w}$ to be the sub-vector of $\bm{v}''$ with
  these extra column entries removed completes the proof.
\qed\end{proof}

Repeated application of the previous lemma allows us to completely
eliminate all of the columns in $\bm{M}'_{\bm\gamma}$ and $\bm{N}_n$
other than the triangular parts, at the cost of having to consider all
possible column-subsets of ${\bm\gamma}$ itself. This leads to the following
condition:
\begin{equation}\label{eqn:safepp}
  \safe''(\bm{M}'_{\bm\gamma},\bm{N}_n) :=\ 
  \forall\, k\in\{1,\ldots,n\},
  \forall\, \bm{P}\in\select^n_k,\ 
  \safe(\tpart(\bm{M}'_{{\bm\gamma}\bm{P}}),
    \tpart(\bm{N}_k)).
\end{equation}
In other words, we restrict our attention to only square submatrices
of the triangular parts of $\bm{M}'_{\bm\gamma}$ and $\bm{N}_n$.
As it turns out, this condition is exactly equivalent to the original
one.

\begin{theorem}\label{thm:safepp}
  For any field $\K$, matrix ${\bm\gamma}\in\K^{(d+1)\times n}$ where $n\ge 1$, and
  matrix $\bm{N}_n\in\{\bm{L}_n,\bm{L}'_n\}$, we have
  $\safe''(\bm{M}'_{\bm\gamma}, \bm{N}_n) \Leftrightarrow
  \safe(\bm{M}'_{\bm\gamma}, \bm{N}_n)$.
\end{theorem}
\begin{proof}
  We prove the equivalent double negation
  $\neg\safe(\bm{M}'_{\bm\gamma},\bm{N}_n) \Leftrightarrow
  \neg\safe''(\bm{M}'_{\bm\gamma},\bm{N}_n)$.

  First we prove the ``$\Rightarrow$'' direction by induction on $n$.
  Assuming that $\neg\safe(\bm{M}'_{\bm\gamma},\bm{N}_n)$ means there
  exists a vector $\bm{v}\in\K^{\ell_n}$ such that
  $\wt(\bm{v})\le n$, $\bm{M}'_{\bm\gamma}\bm{v}=\zm$, and
  $\bm{N}_n\bm{v}$ has full weight $n+1$.

  For the base case, let $n=1$.
  Because $\wt(\bm{v})=1$ and $\wt(\bm{N}_n\bm{v})=2$,
  the lone non-zero entry of $\bm{v}$ must correspond to a weight-2 column
  in $\bm{N}_n$, and the only such columns are in the triangular
  part.
  So considering the vector formed from the last
  $d+1$ entries of $\bm{v}$ shows that
  $\neg\safe(\tpart(\bm{M}'_{\bm\gamma}),\tpart(\bm{N}_n))$,
  which is equivalent to
  $\neg\safe''(\bm{M}'_{\bm\gamma},\bm{N}_n)$
  when $n=1$.

  Now for the induction case, let $n\ge 2$ and
  assume the $\Rightarrow$ direction is true for all size-$(n-1)$
  subsets of columns of ${\bm\gamma}$.

  Again we start with a vector $\bm{v}$ which is a counterexample to
  $\safe(\bm{M}'_{\bm\gamma},\bm{N}_n)$. If $\bm{v}$ has any
  non-zero entry in indices $2$ through $dn+3n+1$, then we can isolate
  that entry in its own vector $\bm{u}$ and write
  $\bm{v}=\bm{u}+\bm{v}^*$, where $\wt(\bm{v}^*) =
  \wt(\bm{v})-1 \le n-1$. Now apply \cref{lem:removecol} to obtain a
  vector $\bm{w}\in\K^{\ell_{n-1}}$ and a selection matrix
  $\bm{P}\in\select^n_{n-1}$ such that
  $\wt(\bm{w})\le n-1$,
  $\bm{M}'_{{\bm\gamma}\bm{P}}\bm{w}=\zm$, and
  $\wt(\bm{N}_{n-1}\bm{w})=n-1$.
  Therefore
  $\neg\safe(\bm{M}'_{{\bm\gamma}\bm{P}},\bm{N}_{n-1})$, so we
  can apply the induction hypothesis to complete this sub-case.

  Otherwise, the non-zero entries of $\bm{v}$ are in the very
  first index, or in the last $(d+1)n$ indices which correspond to the
  triangular parts. But the first columns of $\bm{N}_n$ and
  $\bm{M}'_{\bm\gamma}$ are all zeros except for the first row in
  $\bm{N}_n$, which is eliminated in the triangular part
  $\tpart(\bm{N}_n)$. Therefore, if this entry of $\bm{v}$ is
  non-zero, we can change it to zero without affecting
  $\bm{M}'_{\bm\gamma}\bm{v}$, which must equal $\zm$, or the last
  $n$ rows of $\bm{N}_n\bm{v}$, which must be all non-zero.
  Hence the vector consisting of the last $(d+1)n$ entries of
  $\bm{v}$ is a counterexample to
  $\safe(\tpart(\bm{M}'_{\bm\gamma}),\tpart(\bm{N}_n))$. This
  completes the $\Rightarrow$ direction of the proof.

  For the $\Leftarrow$ direction, assume that
  $\neg\safe''(\bm{M}'_{\bm\gamma},\bm{N}_n)$. This means there is
  some $k\in\{1,\ldots,n\}$, some selection of columns from ${\bm\gamma}$
  defined by $\bm{P}\in\select^n_k$, and some
  $\bm{v}\in\K^{\ell_k}$ such that
  $\wt(\bm{v})\le k$, $\tpart(\bm{M}'_{{\bm\gamma}\bm{P}})\bm{v}=\zm$, and
  $\tpart(\bm{N}_k)\bm{v}$ has full weight $k$.

  Because the triangular part is a subset of the whole, we can prepend
  $\bm{v}$ with $dk+3k+1$ zeros to obtain a vector $\bm{v}'$
  such that $\bm{M}'_{{\bm\gamma}\bm{P}}\bm{v}'=\zm$ and
  $\bm{N}_k\bm{v}'$ is non-zero everywhere except possibly in the
  first row. Observe that the row of $\bm{N}_k$ immediately above
  the triangular part is exactly identical to the top row of
  $\tpart(\bm{N}_k)$, so in fact $\bm{N}_k\bm{v}'$ has full
  weight $k+1$.

  This shows that there exists at least one $k\ge 1$ such that there
  exists a
  selection $\bm{P}\in\select^n_k$ and a vector $\bm{v}'$ which
  is a counterexample to
  $\safe(\bm{M}'_{{\bm\gamma}\bm{P}},\bm{N}_k)$.
  Assume now that $k$ is the \emph{largest} such integer.

  If $k=n$, then $\bm{M}'_{{\bm\gamma}\bm{P}}=\bm{M}'_{\bm\gamma}$,
  and $\bm{v}'$ is a counterexample to
  $\safe(\bm{M}'_{\bm\gamma},\bm{N}_n)$ already.

  Otherwise, if $k<n$, we show that we can construct a larger selection
  matrix $\bm{Q}$ and corresponding vector $\bm{w}$
  satisfying the conditions above, which is a contradiction to the
  assumption that $k$ is the largest such value.

  Construct another selection matrix
  $\bm{Q}\in\select^n_{k+1}$ consisting of the columns selected by
  $\bm{P}$ plus some additional column $i$; for convenience write
  ${\bm\zeta} = {\bm\gamma}\bm{Q}$.
  Note that
  $\bm{M}'_{{\bm\gamma}\bm{P}}$ and $\bm{N}_k$ are submatrices of
  $\bm{M}'_{{\bm\zeta}}$ and $\bm{N}_{k+1}$ respectively,
  the latter both having exactly one more row and some number of extra
  columns. Therefore by extending $\bm{v}'$ to a larger vector
  $\bm{v}''$ by inserting zeros in the locations of these extra
  columns, we have that
  $\bm{M}'_{{\bm\zeta}}\bm{v}''$ is zero everywhere except
  possibly at index $i$, and $\bm{N}_{k+1}\bm{v}''$ is non-zero
  everywhere except at index $i$. Let $a$ be the $i$th entry of
  $\bm{M}'_{{\bm\zeta}}\bm{v}''$ and $b$ be the $i$th
  entry of $\bm{N}_{k+1}\bm{v}''$.

  Finally, we show how to add one more entry to $\bm{v}''$ to
  ``fix'' the exceptions at index $i$ in the previous sentence, making
  $a=0$ and $b\ne 0$. There are four cases to consider:
  \begin{enumerate}
    \item If $a=0$ and $b\ne 0$, then we are done.
    \item If $a=0$ and $b=0$, then set the $(i+1)$th entry of
      $\bm{v}$ to 1; this corresponds to a column of zeros in
      $\bm{M}'_{{\bm\zeta}}$ and a column of the identity
      matrix in $\bm{N}_{k+1}$. So adding that column keeps $a=0$
      but sets $b$ to $1$.
    \item If $a\ne 0$ and $b\ne 0$, then set the $(k+i+1)$th entry of
      $\bm{v}$ to $-a$. This entry corresponds to a column of the
      identity matrix in $\bm{M}'_{\bm\zeta}$ and a column of zeros in
      $\bm{N}_{k+1}$, so adding it keeps $b\ne 0$ but cancels the
      value of $a$.
    \item If $a\ne 0$ and $b=0$, then set the $(2k+i+2)$th entry of
      $\bm{v}$ to $-a/{\bm\zeta}_{0,i}$. This entry corresponds to a
      column of $\diag_{{\bm\zeta},0}$ in $\bm{M}'_{{\bm\zeta}}$, and a
      column of either $\bm{I}_{k+1}$ or $\omega_0\bm{I}_{k+1}$
      within
      $\bm{N}_{k+1}$, and therefore the change to $\bm{v}$
      cancels out $a$ and sets $b$ to some non-zero value.
  \end{enumerate}
  This newly constructed vector has weight at most $\wt(\bm{v}'')+1
  \le k+1$, and is therefore a counterexample to
  $\safe(\bm{M}'_{\bm\zeta},\bm{N}_{k+1})$. This is a contradiction
  to the assumption that $k$ was maximal, which completes the
  $\Leftarrow$ direction and the entire proof.
\qed\end{proof}

\section{An MDS precondition}
\label{sec:precond}

We use the results of the previous two sections to develop a useful
precondition for generating ${\bm\gamma}$ matrices which satisfy the safety
and correctness conditions of the two schemes.
This precondition guarantees the correctness conditions, and (as we will
see in later sections) seems to raise the probability of a matrix satisfying the
safety condition. We then show how to explicitly generate
matrices which satisfy these preconditions.

\subsection{Preconditions intuition and definition}

As in the previous section, let ${\bm\gamma}\in\K^{(d+1)\times d}$ be a
matrix whose entries determine the correctness and safety of one of the
two masking schemes according to \cref{prop:scheme4} or
\cref{prop:scheme5}. (Either ${\bm\gamma}$ must have a row of 1's for the
masking scheme with linear bilinear multiplications, or the columns of
${\bm\gamma}$ must sum to $0$ for the scheme with linear randomness.)

Then \cref{thm:safeprime,thm:safepp} tell us that a sufficient condition
is that, for every
square submatrix of $\tpart(\bm{M}'_{\bm\gamma})$,
every vector in its right kernel results in at least one zero entry when multiplied
by a corresponding submatrix of $\tpart(\bm{N}_d)$. The general
idea of the MDS precondition we describe next is to
\emph{minimize the rank of this right kernel}, effectively limiting the
number of possible ``unsafe'' vectors. In particular, when a
square submatrix of $\tpart(\bm{M}'_{\bm\gamma})$ is non-singular, then
its nullity is zero and the scheme is safe with respect to that subset
of rows and columns.

This suggests a strategy to increase the likelihood of a matrix leading to a safe scheme: one may try to choose ${\bm\gamma}$ in a way that ensures
that $\tpart(\bm{M}'_{{\bm\gamma}\bm{P}})\bm{Q}$ has a trivial kernel for as many
selection matrices $\bm{P}\in\select^d_k$
and $\bm{Q}\in\select^{\ell_k}_{k}$ as possible. That
is, square submatrices of the triangular part of $\bm{M}'_{\bm\gamma}$
should be non-singular as often as possible.

A good such choice for ${\bm\gamma}$ is to take it to be
such that all its square submatrices are MDS.
To justify this claim, recall from \cref{sec:prelims} that
any square submatrix of an MDS matrix is invertible, \ie, has a trivial
kernel.  Further, from the definition of $\tpart(\bm{M}'_{\bm\gamma})$,
its columns consist of (partial) rows of ${\bm\gamma}$; therefore many of
its submatrices are in fact (transposed) submatrices of ${\bm\gamma}$ itself.

\begin{example}\label{ex:mds3}
Consider for the case $d=3$, the submatrix of $\tpart(\bm{M'}_{\bm\gamma})$ given by:
\[
\bm{X} = \begin{pmatrix}
{\bm\gamma}_{0,1} & {\bm\gamma}_{1,1} & {\bm\gamma}_{2,1}\\
0 & {\bm\gamma}_{1,2} & {\bm\gamma}_{2,2}\\
0 & {\bm\gamma}_{1,3} & {\bm\gamma}_{2,3}\\
\end{pmatrix}.
\]
(Note that in the case of \cref{cond41}, ${\bm\gamma}_{0,1}$ must equal 1.)
If all square submatrices of ${\bm\gamma}$ are MDS, the bottom-right $2\times 2$ submatrix of
$\bm{X}$ is necessarily non-singular, and ${\bm\gamma}_{0,1}\ne 0$, so
therefore this entire submatrix is non-singular.
This would not be the case for an arbitrary matrix ${\bm\gamma}$, even if say, one takes it to be full-rank.
\end{example}

We now state our two \emph{preconditions} on the matrices used to instantiate either masking scheme.
As will be clear in the remainder of this paper, these preconditions are by no means sufficient, nor necessary.
Yet we will also see, both formally (in \cref{sec:analytic}) and experimentally (in \cref{sec:exp}) how they may be useful.

\setcounter{pcondsec}{4}\setcounter{precondition}{0}
\begin{precondition}\label{pcond41}
A matrix ${\bm\gamma}\in\K^{(d+1)\times d}$ \emph{satisfies
\cref{pcond41}} for \cref{cond41} if
it can be written as
$\displaystyle {\bm\gamma} = \begin{pmatrix}\om_{1\times
d}\\\bm{A}\end{pmatrix}$, and both matrices
$\bm{A}$ and $\om_{d\times d} - \bm{A}$ are
row XMDS.
\end{precondition}

Any such matrix ${\bm\gamma}$
clearly satisfies the correctness condition, which is item (1) in
\cref{prop:scheme4}. The XMDS property also ensures that all square
submatrices of ${\bm\gamma}$ and ${\bm\delta}$
are non-singular, which (we expect) will make the safety conditions (2)
and (3) from \cref{prop:scheme4} more likely satisfied.

\stepcounter{pcondsec}
\begin{precondition}\label{pcond51}
A matrix ${\bm\gamma} \in \K^{(d+1)\times d}$
\emph{satisfies \cref{pcond51}} for \cref{cond51} if
$\sum_{i=0}^{d} {\bm\gamma}_i = \zm_{1\times d}$ and
all of its square submatrices are MDS.
\end{precondition}

Again, this precondition guarantees the correctness of the
scheme, corresponding to item (1) of \cref{prop:scheme5}, and the
non-singular submatrices make it (we expect) more likely that the
security condition, item (2), is also true.

\subsection{Explicit constructions satisfying the preconditions}

It is relatively easy to check if a given matrix satisfies either of the above preconditions.
Here we do even better, providing a direct construction for families of
matrices that satisfy each precondition.

\begin{theorem}[Satisfying \cref{pcond41}]\label{thm:pcond41}
Let $\{x_1,\ldots,x_d,y_1,\ldots,y_d\} \in \K\backslash\{0\}$ be $2d$
distinct non-zero elements of \K, and define matrix
$\bm{A}\in\K^{d\times d}$ by
$\bm{A}_{i,j} = x_i / (x_i - y_j)$.
Then the corresponding ${\bm\gamma}\in\K^{(d+1)\times d}$
satisfies \cref{pcond41}.
\end{theorem}
\begin{proof}
Define the row-extended Cauchy matrix $\bm{B}$ as $\bm{B}_{0,j} = 1$, $1 \leq j \leq d$;
$\bm{B}_{i,j} = (x_i - y_j)^{-1}$, $1 \leq i, j \leq d$. The
generalized extended matrix
obtained from $\bm{B}$ by the row scaling
$\bm{c} = \begin{pmatrix}1 & x_1 & \cdots & x_d \end{pmatrix}$
is equal to ${\bm\gamma}$,
and all its square submatrices are invertible by construction,
hence $\bm{A}$ is row XMDS.

The matrix
$\bm{C} = \om_{d\times d} - \bm{A}$ is given by $\begin{pmatrix}(x_i - y_j - x_i)\cdot(x_i - y_j)^{-1}\end{pmatrix} =
\begin{pmatrix}-y_j\cdot(x_i - y_j)^{-1}\end{pmatrix}$. It is a generalized Cauchy matrix with column scaling given by
$\begin{pmatrix}-y_1 & \ldots & -y_d\end{pmatrix}^T$, and is then MDS.
Because $0 \notin \{x_1,\ldots,x_d,y_1,\ldots,y_d\}$, one may extend $\bm{C}$ by one row on top using $x_0 = 0$, resulting
in $\bm{C'}$ s.t. $\bm{C'}_{0,j} = -y_j\cdot(0 - y_j)^{-1} = 1$, $1 \leq j \leq d$; $\bm{C'}_{i,j} = \bm{C}_{i,j}$,
$1 \leq i,j \leq d$. In other words,
\[
\bm{C'} = \begin{pmatrix}
\om_{1\times d}\\
\bm{C}
\end{pmatrix}
\]
is a generalized Cauchy matrix, whose square submatrices are all invertible by construction, hence $\bm{C} = \om_{d \times d} -\bm{A}$ is row XMDS.
\qed
\end{proof} 

\begin{theorem}[Satisfying \cref{pcond51}]\label{thm:pcond51}
Let $\{x_1,\ldots,x_d,x_{d+1},y_1,\ldots,y_d\} \in \K$ be $2d+1$ distinct elements of $\K$;
let $\bm{A} = \begin{pmatrix}(x_i - y_j)^{-1}\end{pmatrix}$;
and let $\bm{c} = \begin{pmatrix}c_1 & \cdots & c_{d+1}\end{pmatrix}$
be a non-zero vector in the left kernel of $\bm{A}$. Then ${\bm\gamma} = \begin{pmatrix}c_i\cdot(x_i - y_j)^{-1}\end{pmatrix}$
satisfies \cref{pcond51}.
\end{theorem}
\begin{proof}
By construction, the $d+1\times d$ Cauchy matrix $\bm{A}$ has a left kernel of dimension one.
Furthermore, any vector of this kernel that is not the null vector is of full Hamming weight,
as being otherwise would imply the existence of $k \leq d$ linearly-dependent rows of $\bm{A}$.
The row scaling coefficients $\begin{pmatrix}c_1 & \cdots & c_{d+1}\end{pmatrix}$ are thus all non-zero, and the generalized Cauchy matrix $\bm{A}'$
is such that its rows sum to the null vector and all its square submatrices are invertible.
\qed
\end{proof}

\section{Analytic construction for order up to 3}
\label{sec:analytic}

In this section, we develop explicit polynomial conditions on the
entries of generalized Cauchy matrices that are sufficient to ensure
both the correctness and safety of two schemes described in
\cref{sec:scheme_def}.

The results are explicit constructions for many field sizes. For
order $d=1$, \cref{cor:dim1} proves that any non-zero ${\bm\gamma}$ matrix
makes the scheme secure.
For order $d=2$, \cref{cor:dim2} proves that our MDS preconditions
in the previous section always produce safe constructions without the
need for any further checks. Finally, for order $d=3$,
\cref{thm:dim3explicit41,thm:dim3explicit51} provide $x_i$ and $y_i$
values to use in those MDS preconditions in order to generate safe
constructions for any field of characteristic 2 with $q\ge 4$.

The idea behind our MDS preconditions in \cref{sec:precond} was
to ensure that all square submatrices of ${\bm\gamma}$ are non-singular, and
therefore \emph{many} square submatrices of the matrix
$\tpart(\bm{M}'_{\bm\gamma})$ have nullity zero. For small dimensions,
we can go further and actually require that \emph{all}
submatrices of $\tpart(\bm{M}'_{\bm\gamma})$ are non-singular which could possibly
violate the condition $\safe''$ from \eqref{eqn:safepp}. This will in
turn guarantee a safe and correct construction by
\cref{thm:safepp,prop:scheme4,prop:scheme5}.

\subsection{Columns which must be selected}

Let ${\bm\gamma}\in\K^{(d+1)\times n}$
and recall the definitions of $\tpart(\bm{N}_n)$ and
$\tpart(\bm{M}'_{\bm\gamma})$; in the former case we show only the positions of
the non-zero entries, which are the same whether
$\bm{N}_n=\bm{L}_n$ or $\bm{N}_n=\bm{L}'_n$.

\[\begin{matrix}
\multirow{4}{*}{$\tpart(\bm{N}_n)$} &
\multirow{4}{*}{$=$} &
\multirow{4}{*}{$\displaystyle\left(\rule{0pt}{28pt}\right.$} &
    * & * & \cdots & *      & \rule{10pt}{0pt} &
    * & * & \cdots & *      & \rule{10pt}{0pt} &
    \multirow{4}{*}{$\cdots$} & \rule{10pt}{0pt} &
    * & * & \cdots & *      &
\multirow{4}{*}{$\displaystyle\left.\rule{0pt}{28pt}\right),$} \\
&&&   & * & \cdots & *      & &
      & * & \cdots & *      & &
      &&
      & * & \cdots & *      & \\
&&&   &   & \ddots & \vdots & &
      &   & \ddots & \vdots & &
      &&
      &   & \ddots & \vdots & \\
&&&   &   &        & *      & &
      &   &        & *      & &
      &&
      &   &        & *      &
\\[10pt]
\multirow{4}{*}{$\tpart(\bm{M}'_{\bm\gamma})$} &
\multirow{4}{*}{$=$} &
\multirow{4}{*}{$\displaystyle\left(\rule{0pt}{28pt}\right.$} &
    {\bm\gamma}_{0,1} & {\bm\gamma}_{0,1} & \cdots & {\bm\gamma}_{0,1} & \rule{10pt}{0pt} &
    {\bm\gamma}_{1,1} & {\bm\gamma}_{1,1} & \cdots & {\bm\gamma}_{1,1} & \rule{10pt}{0pt} &
    \multirow{4}{*}{$\cdots$} & \rule{10pt}{0pt} &
    {\bm\gamma}_{d,1} & {\bm\gamma}_{d,1} & \cdots & {\bm\gamma}_{d,1} &
\multirow{4}{*}{$\displaystyle\left.\rule{0pt}{28pt}\right).$} \\
&&&   & {\bm\gamma}_{0,2} & \cdots & {\bm\gamma}_{0,2} & &
      & {\bm\gamma}_{1,2} & \cdots & {\bm\gamma}_{1,2} & &
      &&
      & {\bm\gamma}_{d,2} & \cdots & {\bm\gamma}_{d,2} & \\
&&&   &   & \ddots & \vdots & &
      &   & \ddots & \vdots & &
      &&
      &   & \ddots & \vdots & \\
&&&   &   &        & {\bm\gamma}_{0,n} & &
      &   &        & {\bm\gamma}_{1,n} & &
      &&
      &   &        & {\bm\gamma}_{d,n} &
\end{matrix}
\]

Notice that all pairs of columns in $\bm{M}'_{\bm\gamma}$ and $\bm{N}_n$
with the same index (hence corresponding to the same probe in the
masking scheme) have the same weight.
The next lemma shows that any unsafe set of probes from among these
columns must include at least two of the full-weight columns.

\begin{lemma}\label{lem:fullcols}
  Let ${\bm\gamma}\in\K^{(d+1)\times n}, \bm{M}'_{\bm\gamma}, \bm{L}_n$
  be as above.
  If ${\bm\gamma}$ has no zero entries, then
  any column selection $\bm{P}\in\select^{\ell_n}_n$ which is a
  counterexample to
  $\safe'(\tpart(\bm{M}'_{\bm\gamma}),\tpart(\bm{N}_n))$ must include
  at least two columns of full weight $n$ from
  $\tpart(\bm{M}'_{\bm\gamma})$ and $\tpart(\bm{N}_n)$.
\end{lemma}
\begin{proof}
  A counterexample to
  $\safe'(\tpart(\bm{M}'_{\bm\gamma}),\tpart(\bm{N}_n))$
  is a selection matrix $\bm{P}\in\select^{\ell_n}_n$
  such that the matrix product
  $\tpart(\bm{N}_n)\bm{P}\cdot\kbasis(\tpart(\bm{M}'_{\bm\gamma})\bm{P})$
  has no zero rows.

  The only columns of $\tpart(\bm{N}_n)$ which are
  non-zero in the last row are those columns of full weight, so at least
  one must be included in $\bm{P}$ for the product to have no zero
  rows. But in order for $\tpart(\bm{M}'_{\bm\gamma})\bm{P}$ to have
  a non-trivial kernel, it must have a \emph{second} column with a
  non-zero in the last row.
\qed\end{proof}

\subsection{Dimensions 1 and 2}

Combined with the results of the prior sections,
this leads immediately to solutions for orders $n=1$ or $n=2$.

\begin{corollary}\label{cor:dim1}
  For any ${\bm\gamma}\in\K^{(d+1)\times 1}$ that contains no zero entries,
  we have $\safe(\bm{M}'_{\bm\gamma}, \bm{N}_1)$.
\end{corollary}
\begin{proof}
  Clearly there is no way to include two full-weight columns in a
  selection $\bm{P}\in\select^{\ell_1}_1$ of a single column.
  Therefore from \cref{lem:fullcols}, we have
  $\neg\safe'(\tpart(\bm{M}'_{\bm\gamma}),\tpart(\bm{N}_1))$.
  By \cref{thm:safeprime,thm:safepp} this implies the statement above.
\qed\end{proof}

\begin{corollary}\label{cor:dim2}
  For any ${\bm\gamma}\in\K^{(d+1)\times 2}$ such that all square submatrices
  of ${\bm\gamma}$ are MDS, we have
  $\safe(\bm{M}'_{\bm\gamma},\bm{N}_2)$.
\end{corollary}
\begin{proof}
  Any selection of 2 columns of $\tpart(\bm{M}'_{\bm\gamma})$
  that includes at least 2 full-weight
  columns is simply a transposed submatrix of ${\bm\gamma}$ of dimension 2.
  By \cref{prop:mds_minors}, any such submatrix is non-singular, and
  thus has a trivial kernel. Therefore by \cref{lem:fullcols}
  there are no counterexamples to
  $\safe'(\tpart(\bm{M}'_{\bm\gamma}),\tpart(\bm{N}_2))$,
  and by \cref{thm:safeprime,thm:safepp} again the stated result
  follows.
\qed\end{proof}

Most notably, these corollaries guarantee that \emph{any} matrix with
column dimension 1 or 2 which satisfies \cref{pcond41} or \cref{pcond51}
is an instantiation of the respective masking scheme that is correct and
safe. Because we have explicit constructions for these preconditions in
\cref{thm:pcond41,thm:pcond51} over any field $\Fq$ with $q>2d+1$, we
also have explicit instantiations for the masking schemes secure against
1 or 2 probes.

\subsection{Dimension 3}

Next we turn to the case of $n=3$. It is no longer possible to construct
safe instances of ${\bm\gamma}$ based on the MDS preconditions alone, but
there is only one other shape of square submatrices that need be
considered.

\begin{lemma}\label{lem:dim3}
  Let ${\bm\gamma}\in\K^{(d+1)\times 3}, \bm{M}'_{\bm\gamma}, \bm{L}_n$
  be as above.
  If every square submatrix of ${\bm\gamma}$ is MDS, and for all distinct
  triples of indices $\{i,j,k\}\subseteq \{0,1,\ldots,d+1\}$ the matrix
  \[
    \begin{pmatrix}
      {\bm\gamma}_{i,1} & {\bm\gamma}_{j,1} & {\bm\gamma}_{k,1} \\
      {\bm\gamma}_{i,2} & {\bm\gamma}_{j,2} & {\bm\gamma}_{k,2} \\
      {\bm\gamma}_{i,3} & {\bm\gamma}_{j,3} & 0
    \end{pmatrix}
  \]
  is non-singular, then we have
  $\safe(\bm{M}'_{\bm\gamma},\bm{N}_3)$.
\end{lemma}
\begin{proof}
  The goal is to ensure that no square submatrix of
  $\tpart(\bm{M}'_{\bm\gamma})$ which could possibly be part of a
  counterexample to
  $\safe'(\tpart(\bm{M}'_{\bm\gamma}),\tpart(\bm{N}_3))$ has a
  non-trivial kernel. Already we know from \cref{lem:fullcols} that any
  such submatrix must include two distinct full-weight columns.
  Because all square submatrices of ${\bm\gamma}$ are MDS, these two columns
  have a trivial kernel, meaning a third column must be added if one
  hopes to find a counterexample. This leads to three cases, depending
  on the weight of this third column.

  If the third column has weight 1, the situation is analogous to that
  of \cref{ex:mds3}. The corresponding matrix is non-singular if and only
  if some $2\times 2$ submatrix of ${\bm\gamma}$ is non-singular, which it
  must be by the MDS assumption.

  Next, if the third column has full weight 3, then we have a $3\times
  3$ submatrix of ${\bm\gamma}$, which again must be non-singular.

  The remaining case is that the third column has weight 2, as in the
  statement of the lemma. All that remains is to prove that this index
  $k$ must be distinct from $i$ and $j$. By way of contradiction, and
  without loss of generality, suppose $i=k$. Then after subtracting the
  third column from the first, we obtain the matrix
  \[\begin{pmatrix}
      0 & {\bm\gamma}_{j,1} & {\bm\gamma}_{i,1} \\
      0 & {\bm\gamma}_{j,2} & {\bm\gamma}_{i,2} \\
      {\bm\gamma}_{i,3} & {\bm\gamma}_{j,3} & 0
  \end{pmatrix},\]
  which is non-singular if and only if the original matrix is
  non-singular. And indeed, this matrix must be non-singular because the
  upper-right $2\times 2$ matrix is a submatrix of ${\bm\gamma}$.

  Therefore the only remaining case of a submatrix which could be a
  counterexample to
  $\safe'(\tpart(\bm{M}'_{\bm\gamma}),\tpart(\bm{N}_3))$ is one of
  the form given in the statement of the lemma. Applying once again
  \cref{thm:safeprime,thm:safepp} completes the proof.
\qed\end{proof}

This finally leads to a way to construct safe instances for the schemes
when $d=3$ based only on polynomial conditions, via the following steps:
\begin{enumerate}
  \item Write down a symbolic $4\times 3$ matrix ${\bm\gamma}$ satisfying
    \cref{pcond41} or \cref{pcond51} according to the constructions of
    \cref{thm:pcond41} (resp.\ \cref{thm:pcond51}), leaving all the
    $x_i$'s and $y_i$'s as indeterminates.
  \item Extract all $3\times 3$ matrices from ${\bm\gamma}$ that match the
    form of \cref{lem:dim3} and compute their determinants, which are
    rational functions in the $x_i$s and $y_i$s.
  \item Factor the numerators of all determinants, removing duplicate
    factors and factors such as $x_i-y_i$ which must be non-zero by
    construction.
  \item A common non-root to the resulting list of polynomials
    corresponds to a ${\bm\gamma}$ matrix which is safe for the given scheme.
\end{enumerate}

Next we show the results of these computations for each of the two
schemes. We used the Sage~\cite{sagemath} computer algebra system to compute
the lists of polynomials according to the procedure above, which takes
about 1 second on a modern laptop computer.

\begin{figure}[bp]
\[\boxed{\scalebox{0.75}{$\displaystyle\begin{matrix}
x_{2} x_{3} -  y_{1} y_{2} -  x_{2} y_{3} -  x_{3} y_{3} + y_{1} y_{3} + y_{2} y_{3} \\
x_{2} x_{3} -  x_{3} y_{1} -  x_{3} y_{2} + y_{1} y_{2} -  x_{2} y_{3} + x_{3} y_{3} \\
x_{2} x_{3} -  x_{2} y_{1} -  x_{2} y_{2} + y_{1} y_{2} + x_{2} y_{3} -  x_{3} y_{3} \\
x_{1} x_{3} -  y_{1} y_{2} -  x_{1} y_{3} -  x_{3} y_{3} + y_{1} y_{3} + y_{2} y_{3} \\
x_{1} x_{3} -  x_{3} y_{1} -  x_{3} y_{2} + y_{1} y_{2} -  x_{1} y_{3} + x_{3} y_{3} \\
x_{1} x_{3} -  x_{1} y_{1} -  x_{1} y_{2} + y_{1} y_{2} + x_{1} y_{3} -  x_{3} y_{3} \\
x_{1} x_{2} -  y_{1} y_{2} -  x_{1} y_{3} -  x_{2} y_{3} + y_{1} y_{3} + y_{2} y_{3} \\
x_{1} x_{2} -  x_{2} y_{1} -  x_{2} y_{2} + y_{1} y_{2} -  x_{1} y_{3} + x_{2} y_{3} \\
x_{1} x_{2} -  x_{1} y_{1} -  x_{1} y_{2} + y_{1} y_{2} + x_{1} y_{3} -  x_{2} y_{3} \\
x_{2} y_{1} y_{2} -  x_{3} y_{1} y_{2} -  x_{2} x_{3} y_{3} + x_{3} y_{1} y_{3} + x_{3} y_{2} y_{3} -  y_{1} y_{2} y_{3} \\
x_{2} y_{1} y_{2} -  x_{3} y_{1} y_{2} + x_{2} x_{3} y_{3} -  x_{2} y_{1} y_{3} -  x_{2} y_{2} y_{3} + y_{1} y_{2} y_{3} \\
x_{1} y_{1} y_{2} -  x_{3} y_{1} y_{2} -  x_{1} x_{3} y_{3} + x_{3} y_{1} y_{3} + x_{3} y_{2} y_{3} -  y_{1} y_{2} y_{3} \\
x_{1} y_{1} y_{2} -  x_{3} y_{1} y_{2} + x_{1} x_{3} y_{3} -  x_{1} y_{1} y_{3} -  x_{1} y_{2} y_{3} + y_{1} y_{2} y_{3} \\
x_{1} y_{1} y_{2} -  x_{2} y_{1} y_{2} -  x_{1} x_{2} y_{3} + x_{2} y_{1} y_{3} + x_{2} y_{2} y_{3} -  y_{1} y_{2} y_{3} \\
x_{1} y_{1} y_{2} -  x_{2} y_{1} y_{2} + x_{1} x_{2} y_{3} -  x_{1} y_{1} y_{3} -  x_{1} y_{2} y_{3} + y_{1} y_{2} y_{3} \\
x_{2} x_{3} y_{1} + x_{2} x_{3} y_{2} -  x_{2} y_{1} y_{2} -  x_{3} y_{1} y_{2} -  x_{2} x_{3} y_{3} + y_{1} y_{2} y_{3} \\
x_{1} x_{3} y_{1} + x_{1} x_{3} y_{2} -  x_{1} y_{1} y_{2} -  x_{3} y_{1} y_{2} -  x_{1} x_{3} y_{3} + y_{1} y_{2} y_{3} \\
x_{1} x_{2} y_{1} + x_{1} x_{2} y_{2} -  x_{1} y_{1} y_{2} -  x_{2} y_{1} y_{2} -  x_{1} x_{2} y_{3} + y_{1} y_{2} y_{3} \\
x_{1} x_{2} x_{3} -  x_{2} x_{3} y_{1} -  x_{2} x_{3} y_{2} -  x_{1} y_{1} y_{2} + x_{2} y_{1} y_{2} + x_{3} y_{1} y_{2} -  x_{1} x_{2} y_{3} -  x_{1} x_{3} y_{3} + x_{2} x_{3} y_{3} + x_{1} y_{1} y_{3} + x_{1} y_{2} y_{3} -  y_{1} y_{2} y_{3} \\
x_{1} x_{2} x_{3} -  x_{1} x_{3} y_{1} -  x_{1} x_{3} y_{2} + x_{1} y_{1} y_{2} -  x_{2} y_{1} y_{2} + x_{3} y_{1} y_{2} -  x_{1} x_{2} y_{3} + x_{1} x_{3} y_{3} -  x_{2} x_{3} y_{3} + x_{2} y_{1} y_{3} + x_{2} y_{2} y_{3} -  y_{1} y_{2} y_{3} \\
x_{1} x_{2} x_{3} -  x_{1} x_{2} y_{1} -  x_{1} x_{2} y_{2} + x_{1} y_{1} y_{2} + x_{2} y_{1} y_{2} -  x_{3} y_{1} y_{2} + x_{1} x_{2} y_{3} -  x_{1} x_{3} y_{3} -  x_{2} x_{3} y_{3} + x_{3} y_{1} y_{3} + x_{3} y_{2} y_{3} -  y_{1} y_{2} y_{3}
\end{matrix}$}}\]
  \caption{Polynomials which should be non-zero to generate a safe
  construction according to \cref{cond41}.
  There are 9 degree-2 polynomials with 6 terms, 9 degree-3 polynomials
  with 6 terms, and 3 degree-3 polynomials with 12 terms.%
  \label{fig:polys41}}
\end{figure}

\begin{proposition}
If $x_1,x_2,x_3,y_1,y_2,y_3\in \Fq$ are distinct non-zero elements so
that the list of polynomials in \cref{fig:polys41} all evaluate to non-zero values,
then the matrix ${\bm\gamma}$ constructed
according to \cref{thm:pcond41} generates a safe masking
scheme according to \cref{cond41}.
\end{proposition}

From the degrees of these polynomials, and by the Schwartz-Zippel lemma
\cite{DBLP:journals/jacm/Schwartz80} and applying the union bound,
a safe construction for \cref{cond41} exists over any field $\Fq$ with $q>54$.

In fact, we have an explicit construction for any binary field
$\Fq$ with $q\ge 16$.

\begin{theorem}\label{thm:dim3explicit41}
  Let $(x_1,x_2,x_3) = (
\texttt{1}, \texttt{3}, \texttt{5})$ and
  $(y_1,y_2,y_3)=(\texttt{6}, \texttt{4}, \texttt{a})$.
  Then for any $k\ge 4$, the matrix ${\bm\gamma}$ constructed according to
  \cref{thm:pcond41} generates a safe masking scheme over
  $\mathbb{F}_{2^k}$ according to \cref{cond41}.
\end{theorem}
\begin{proof}
  Small cases with $4\le k \le 8$ are checked computationally by making
  the appropriate substitutions into the polynomials of
  \cref{fig:polys41}.

  For $k\ge 9$, consider the degrees of the $x_i$s and $y_i$s when
  treated as polynomials over $\mathbb{F}_2$. The highest degree is
  $\deg y_3=3$, and all other elements have degree at most $2$.
  Inspecting the polynomials in \cref{fig:polys41}, we see that they are
  all sums of products of at most three distinct variables. Therefore, when evaluated at
  these $x_i$s and $y_i$s, the degree of any resulting polynomial is at
  most $7$. Over $\mathbb{F}_{2^k}$ where $k\ge 8$ there is therefore no
  reduction, and the polynomials are guaranteed to be non-zero in all
  cases because they are non-zero over $\mathbb{F}_{2^8}$.
\qed\end{proof}

Next we do the same for the masking scheme with linear randomness,
namely that of \cref{cond51}.

\begin{figure}[bp]
\[\boxed{\scalebox{0.75}{$\displaystyle\begin{matrix}
x_{2} x_{3} x_{4} -  x_{3} x_{4} y_{1} -  x_{3} x_{4} y_{2} -  x_{2} y_{1} y_{2} + x_{3} y_{1} y_{2} + x_{4} y_{1} y_{2} -  x_{2} x_{3} y_{3} -  x_{2} x_{4} y_{3} + x_{3} x_{4} y_{3} + x_{2} y_{1} y_{3} + x_{2} y_{2} y_{3} -  y_{1} y_{2} y_{3} \\
x_{2} x_{3} x_{4} -  x_{2} x_{4} y_{1} -  x_{2} x_{4} y_{2} + x_{2} y_{1} y_{2} -  x_{3} y_{1} y_{2} + x_{4} y_{1} y_{2} -  x_{2} x_{3} y_{3} + x_{2} x_{4} y_{3} -  x_{3} x_{4} y_{3} + x_{3} y_{1} y_{3} + x_{3} y_{2} y_{3} -  y_{1} y_{2} y_{3} \\
x_{2} x_{3} x_{4} -  x_{2} x_{3} y_{1} -  x_{2} x_{3} y_{2} + x_{2} y_{1} y_{2} + x_{3} y_{1} y_{2} -  x_{4} y_{1} y_{2} + x_{2} x_{3} y_{3} -  x_{2} x_{4} y_{3} -  x_{3} x_{4} y_{3} + x_{4} y_{1} y_{3} + x_{4} y_{2} y_{3} -  y_{1} y_{2} y_{3} \\
x_{1} x_{3} x_{4} -  x_{3} x_{4} y_{1} -  x_{3} x_{4} y_{2} -  x_{1} y_{1} y_{2} + x_{3} y_{1} y_{2} + x_{4} y_{1} y_{2} -  x_{1} x_{3} y_{3} -  x_{1} x_{4} y_{3} + x_{3} x_{4} y_{3} + x_{1} y_{1} y_{3} + x_{1} y_{2} y_{3} -  y_{1} y_{2} y_{3} \\
x_{1} x_{3} x_{4} -  x_{1} x_{4} y_{1} -  x_{1} x_{4} y_{2} + x_{1} y_{1} y_{2} -  x_{3} y_{1} y_{2} + x_{4} y_{1} y_{2} -  x_{1} x_{3} y_{3} + x_{1} x_{4} y_{3} -  x_{3} x_{4} y_{3} + x_{3} y_{1} y_{3} + x_{3} y_{2} y_{3} -  y_{1} y_{2} y_{3} \\
x_{1} x_{3} x_{4} -  x_{1} x_{3} y_{1} -  x_{1} x_{3} y_{2} + x_{1} y_{1} y_{2} + x_{3} y_{1} y_{2} -  x_{4} y_{1} y_{2} + x_{1} x_{3} y_{3} -  x_{1} x_{4} y_{3} -  x_{3} x_{4} y_{3} + x_{4} y_{1} y_{3} + x_{4} y_{2} y_{3} -  y_{1} y_{2} y_{3} \\
x_{1} x_{2} x_{4} -  x_{2} x_{4} y_{1} -  x_{2} x_{4} y_{2} -  x_{1} y_{1} y_{2} + x_{2} y_{1} y_{2} + x_{4} y_{1} y_{2} -  x_{1} x_{2} y_{3} -  x_{1} x_{4} y_{3} + x_{2} x_{4} y_{3} + x_{1} y_{1} y_{3} + x_{1} y_{2} y_{3} -  y_{1} y_{2} y_{3} \\
x_{1} x_{2} x_{4} -  x_{1} x_{4} y_{1} -  x_{1} x_{4} y_{2} + x_{1} y_{1} y_{2} -  x_{2} y_{1} y_{2} + x_{4} y_{1} y_{2} -  x_{1} x_{2} y_{3} + x_{1} x_{4} y_{3} -  x_{2} x_{4} y_{3} + x_{2} y_{1} y_{3} + x_{2} y_{2} y_{3} -  y_{1} y_{2} y_{3} \\
x_{1} x_{2} x_{4} -  x_{1} x_{2} y_{1} -  x_{1} x_{2} y_{2} + x_{1} y_{1} y_{2} + x_{2} y_{1} y_{2} -  x_{4} y_{1} y_{2} + x_{1} x_{2} y_{3} -  x_{1} x_{4} y_{3} -  x_{2} x_{4} y_{3} + x_{4} y_{1} y_{3} + x_{4} y_{2} y_{3} -  y_{1} y_{2} y_{3} \\
x_{1} x_{2} x_{3} -  x_{2} x_{3} y_{1} -  x_{2} x_{3} y_{2} -  x_{1} y_{1} y_{2} + x_{2} y_{1} y_{2} + x_{3} y_{1} y_{2} -  x_{1} x_{2} y_{3} -  x_{1} x_{3} y_{3} + x_{2} x_{3} y_{3} + x_{1} y_{1} y_{3} + x_{1} y_{2} y_{3} -  y_{1} y_{2} y_{3} \\
x_{1} x_{2} x_{3} -  x_{1} x_{3} y_{1} -  x_{1} x_{3} y_{2} + x_{1} y_{1} y_{2} -  x_{2} y_{1} y_{2} + x_{3} y_{1} y_{2} -  x_{1} x_{2} y_{3} + x_{1} x_{3} y_{3} -  x_{2} x_{3} y_{3} + x_{2} y_{1} y_{3} + x_{2} y_{2} y_{3} -  y_{1} y_{2} y_{3} \\
x_{1} x_{2} x_{3} -  x_{1} x_{2} y_{1} -  x_{1} x_{2} y_{2} + x_{1} y_{1} y_{2} + x_{2} y_{1} y_{2} -  x_{3} y_{1} y_{2} + x_{1} x_{2} y_{3} -  x_{1} x_{3} y_{3} -  x_{2} x_{3} y_{3} + x_{3} y_{1} y_{3} + x_{3} y_{2} y_{3} -  y_{1} y_{2} y_{3}
\end{matrix}$}}\]
  \caption{Polynomials which should be non-zero to generate a safe
  construction according to \cref{cond51}.
  There are 12 degree-3 polynomials with 12 terms each.%
  \label{fig:polys51}}
\end{figure}

\begin{proposition}
If $x_1,x_2,x_3,x_4,y_1,y_2,y_3\in \Fq$ are distinct non-zero elements so
that the list of polynomials in \cref{fig:polys51} all evaluate to non-zero values,
then the matrix constructed according to \cref{thm:pcond51} generates a safe masking
scheme according to \cref{cond51}.
\end{proposition}

Applying the
Schwartz-Zippel lemma and union bound in this context guarantees
a safe construction for \cref{cond51} over any field $\Fq$ with $q>36$.
Again, we have an explicit construction for binary fields of order at
least $16$.

\begin{theorem}\label{thm:dim3explicit51}
  Let $(x_1,x_2,x_3,x_4) = (
\texttt{1}, \texttt{2}, \texttt{5}, \texttt{6})$ and
  $(y_1,y_2,y_3)=(\texttt{4}, \texttt{7}, \texttt{f})$.
  Then for any $k\ge 4$, the matrix ${\bm\gamma}$ constructed according to
  \cref{thm:pcond51} generates a safe masking scheme over
  $\mathbb{F}_{2^k}$ according to \cref{cond51}.
\end{theorem}
The proof is the same as \cref{thm:dim3explicit41}, consisting of
computational checks for $4\le k\le 8$ and then an argument for all
$k\ge 9$ based on the degrees of the $x_i$ and $y_i$ polynomials.

\section{Efficient algorithms to test safeness}
\label{sec:algo}

To test whether a matrix may
be used to safely instantiate either of the masking schemes of Belaïd
\etal, we use the condition $\safe'(\bm{M}'_{\bm\gamma},\bm{N}_d)$
defined in \eqref{eqn:safeprime}, which according to \cref{thm:safeprime} is a
sufficient condition for the scheme under consideration to be safe.
The definition of this condition immediately indicates an algorithm,
which we have implemented with some optimizations using M4RIE
\cite{m4rie} for the finite field arithmetic.

\subsection{The algorithm}

To test whether a
matrix ${\bm\gamma}\in\K^{(d+1)\times d}$ satisfies the conditions of
\cref{prop:scheme4} or \cref{prop:scheme5}, simply
construct
$\bm{M}'_{\bm\gamma}$ and $\bm{N}_d$ and for
all $d$-subsets of columns $\bm{P}\in\select^{\ell}_d$,
check if
$\zr(\bm{N}_d\bm{P}\cdot\kbasis(\bm{M}'_{\bm\gamma}\bm{P}))$.

This algorithm is much more efficient than the one directly suggested by \cref{cond41}: instead of testing all $\sum_{i=1}^d\binom{\ell}{i}q^i$ vectors of $\Fq^\ell$ of weight $d$ or less, it is enough to do
$\binom{\ell}{d}$ easy linear algebra computations. While this remains exponential in $d$, it removes the practically insuperable factor $q^d$ and gives a complexity that does not depend on the field size (save for the cost of arithmetic).

(Note that we could have used the condition $\safe''$ as in
\cref{thm:safepp} instead, but this turns out to be more complicated in
practice due to the need to take arbitrary subsets of the rows and
columns of $\bm{M}'_{\bm\gamma}$ and $\bm{N}_d$.)

We now describe two implementation strategies for this algorithm.

\subsection{Straightforward implementation with optimizations}

Two simple optimizations may be used to make a straightforward implementation of the above algorithm more efficient in practice.

\paragraph{Skipping bad column picks.}
We can see already from the support of $\bm{N}_d$
that some subsets of columns $\bm{P}\in\select^\ell_d$ never need to
be checked because $\zr(\bm{N}_d\bm{P})$ is already true,
independent of the actual choice of ${\bm\gamma}$. This is the case for
example when the columns selected by $\bm{P}$ are all of weight 1.

For the specific cases of $d=4$, this reduces the number of supports to be considered from $\binom{49}{4} = 211\,876$ to $103\,030$, saving roughly a factor 2. A similar behaviour is observed for $d=5$, when
one only has to consider $6\,448\,239$ supports among the $\binom{71}{5} = 13\,019\,909$ possible ones.
Note that the same optimization could be applied to the naïve algorithm that exhaustively enumerates low-weight vectors of $\Fq^\ell$.

\paragraph{Testing critical cases first.}
Looking again at how $\bm{M}'_{\bm\gamma}$ is defined, it is easy to see
that for some column selections $\bm{P}$,
$\bm{M}'_{\bm\gamma}\bm{P}$ does not in
fact depend on ${\bm\gamma}$. For these, it is enough to check once
and for all that
$\zr(\bm{N}_{\bm\gamma}\bm{P}\cdot\kbasis(\bm{M}'_{\bm\gamma}\bm{P}))$
indeed holds (if it does not, the scheme would be generically broken). Going
further, even some column subsets such that $\bm{M}_{\bm\gamma}\bm{P}$
actually depends on ${\bm\gamma}$ may always be ``safe'' provided that ${\bm\gamma}$ satisfies a certain precondition, such as
for instance being MDS as suggested in \cref{sec:precond}.

Conversely, it may be the case that for some $\bm{P}$,
$\zr(\bm{N}_d\bm{P}\cdot\kbasis(\bm{M}'_{\bm\gamma}\bm{P}))$
often does
\emph{not} hold. It may then be beneficial to test this subset
$\bm{P}$ before others
that are less likely to make the condition fail. We have experimentally
observed that such subsets do exist. For instance, in the case $d = 5$
for \cref{cond41}, only $\approx 320\,000$ column subsets seem to
determine whether a matrix satisfies the condition or not.\footnote{This figure was found experimentally by regrouping the supports in clusters of $10\,000$, independently of $q$. A more careful analysis may lead to a more precise result.}
There, checking these supports first and using an early-abort strategy,
verifying that a matrix \emph{does not} satisfy the condition is at
least $\approx 20$ times faster than enumerating all possible column
subsets.

\subsection{Batch implementation}

Especially when the matrix ${\bm\gamma}$ under consideration actually
satisfies the required conditions, checking these using the
straightforward strategy entails considerable redundant computation due
to the overlap between subsets of columns.

To avoid this, we also implemented a way to check the condition
$\safe'(\bm{M}'_{\bm\gamma},\bm{N}_d)$ that operates over the entire
matrix simultaneously, effectively considering many subsets of columns
in a single batch.

Recall that the algorithm needs to (1) extract a subset of columns of
$\bm{M}'_{\bm\gamma}$, (2) compute a right kernel basis for this subset,
(3) multiply $\bm{N}_d$ times this kernel basis, and (4) check for
zero rows in the resulting product.

Steps (2) and (3) would typically be performed via Gaussian elimination:
For each column of $\bm{M}'_{\bm\gamma}$ that is in the selection,
we search for a pivot row, permute rows if necessary to
move the pivot up, then eliminate above and below the pivot and move on.
If there is no pivot in some column, this means a new null vector has
been found; we use the previous pivots to compute the null vector and
add it to the basis. Finally, we multiply this null space basis by the
corresponding columns in $\bm{N}_d$ and check for zero rows.

The key observation for this algorithm is that we can perform these
steps (2) and (3) \emph{in parallel} to add one more column to an
existing column selection. That is, starting with some subset of
columns, we consider the effect on the null space basis and the
following multiplication by $\bm{N}_d$ simultaneously for all other
columns in the matrices. Adding columns with pivots does not change the
null space basis or the product with $\bm{N}_d$.
Columns with no pivots add one additional column to the null space
basis, which results in a new column in the product with $\bm{N}_d$.
This new column of
$\bm{N}_d\bm{P}\cdot\kbasis(\bm{M}'_{\bm\gamma}\bm{P})$ may
be checked for non-zero entries and then immediately discarded as the
search continues; in later steps, the \emph{rows} of this product which already have a
non-zero entry no longer need to be considered.

All of this effectively reduces the cost of the check by a factor of
$\ell$ compared to the prior version, replacing the search over all
size-$d$ subsets with a search over size-$(d-1)$ subsets and some matrix
computations. This strategy is especially effective when the ${\bm\gamma}$
matrix under consideration is (nearly or actually) safe, meaning that the early termination
techniques above will not be very useful.

\section{Experimental results and explicit instantiations}
\label{sec:exp}

We implemented both algorithms of the previous section in the practically-useful case of binary fields, using M4RIE for the underlying linear algebra~\cite{m4rie}, and searched for matrices fulfilling \cref{cond41,cond51} in various settings, leading to instantiations of the masking schemes of Belaïd \etal up to $d=6$
and $\mathbb{F}_{2^{16}}$.\footnote{$\mathbb{F}_{2^{16}}$ is the largest
field size implemented in M4RIE, and $d=6$ the maximum dimension for
which safe instantiations (seem to) exist below this field size
limitation.}
We also collected statistics about the fraction of matrices satisfying the conditions, notably in function of the field over which they are defined. This allows to verify experimentally that Precondition~4
of \cref{sec:precond} is useful.

\subsection{Statistics}

We give detailed statistics about the proportion of preconditioned matrices allowing to instantiate either masking scheme up to order 6; this is presented in \cref{tbl:stats_4,tbl:stats_4_2}.
The data was collected by drawing at random matrices satisfying \cref{pcond41} or \cref{pcond51} and checking if they satisfied the safety conditions or not for the respective schemes.

For combinations of field size and order where no safe matrix was found, we give the result as an upper bound.

Notice that the probability for \cref{cond51} appears to be consistently
a bit higher than that for \cref{cond41}. The combinations of field size
$q$ and order $d$ where safe instances are found were almost the same
for both schemes, except that for order 5 and $q=2^9$, where a safe
preconditioned matrix was found for \cref{cond51} but not for
\cref{cond41}. This difference between the schemes can be explained by the fact that
\cref{cond41} places conditions on two matrices ${\bm\gamma}$ and
$\om_{d\times d}-{\bm\gamma}$,
whereas \cref{cond51} depends only on the single matrix ${\bm\gamma}$.

An important remark is that for the smallest field $\mathbb{F}_{2^5}$, the statistics do not include results about the \emph{non-preconditioned safe matrices}, which were the only safe ones we found, see the further discussion below.

We indicate the sample sizes used to obtain each result, as they may vary by several orders of magnitude due to the exponentially-increasing cost of our algorithm with the order. As an illustration,
our batch implementation is able to check 1\,000\,000 dimension-4 matrices over $\mathbb{F}_{2^6}$ in 12\,400 seconds on one core of a 2\,GHz Sandy Bridge CPU, which increases to 590\,000 and 740\,000
seconds for $\mathbb{F}_{2^{12}}$ and $\mathbb{F}_{2^{16}}$ respectively because of more expensive field operations; 1\,600\,000 seconds allowed to test $\approx 145\,000$ and $\approx 25\,000$
dimension-5 matrices for these last two fields, and $\approx 2\,400$ dimension-6 matrices for $\mathbb{F}_{2^{16}}$.

\begin{table}[!htb]
\caption{\label{tbl:stats_4} Instantiations over $\mathbb{F}_{2^5} \sim
\mathbb{F}_{2^{10}}$. Sample sizes (as indicated by symbols in the exponents)
were as follows:  $\ast \approx 400\,000$; $\ddagger = 1\,000\,000$;
$\star \approx 4\,000\,000$; $\dagger \approx 11\,000\,000$.}
\begin{tabu} to \textwidth {lXXXXXX}
\toprule
$q$\hspace{1em} & $2^5$ & $2^6$ & $2^7$ & $2^8$ & $2^9$ & $2^{10}$\\
\midrule
$d$&\multicolumn{6}{c}{\Cref{cond41} \& \Cref{pcond41}} \\
4 & $\leq 2^{-28.8}$ & $2^{-15.25\dagger}$ & $0.009^\dagger$ & $0.11^\ddagger$ & $0.34^\ddagger$ & $0.59^\ddagger$\\
5 & --- & --- & --- & --- & $\leq 2^{-27.5}$ & $2^{-18.9\star}$ \\
\midrule
$d$&\multicolumn{6}{c}{\Cref{cond51} \& \Cref{pcond51}} \\
4 & $\leq 2^{-33.5}$ & $2^{-9.10\ddagger}$ & $0.062^\ddagger$ & $0.27^\ddagger$ & $0.53^\ddagger$ & $0.73^\ddagger$\\
5 & --- & --- & --- & --- & $2^{-18.6\ast}$ & $2^{-11.0\ast}$ \\
\bottomrule
\end{tabu}
\end{table}

\begin{table}[!htb]
\caption{\label{tbl:stats_4_2}
Instantiations over $\mathbb{F}_{2^{11}} \sim \mathbb{F}_{2^{16}}$. Sample sizes (as indicated by symbols in the exponents) were as follows:
$\ddagger = 1\,000\,000$; $\ast \approx 400\,000$; $\diamond \approx 145\,000$; $\bullet \approx 65\,000$; $\triangleleft \approx 40\,000$; $\oslash \approx 30\,000$; $\ltimes \approx 25\,000$; $\wr \approx 560\,000$; $\curlywedge \approx 12\,700$.}
\begin{tabu} to \textwidth {lXXXXXX}
\toprule
$q$\hspace{1em} & $2^{11}$ & $2^{12}$ & $2^{13}$ & $2^{14}$ & $2^{15}$ & $2^{16}$\\
\midrule
$d$&\multicolumn{6}{c}{\Cref{cond41} \& \Cref{pcond41}} \\
4 & $0.77^\ddagger$ & $0.88^\ddagger$ & $0.94^\ddagger$ & $0.97^\ddagger$ & $0.98^\ddagger$ & $0.99^\ddagger$\\
5 & 0.0015$^\ast$ & $0.04^\diamond$ & $0.2^\bullet$ & $0.45^\triangleleft$ & $0.67^\oslash$ & $0.82^\ltimes$\\
6 & --- & --- & --- & --- & $2^{-16.8\wr}$ & $0.003^\curlywedge$\\
\midrule
$d$&\multicolumn{6}{c}{\Cref{cond51} \& \Cref{pcond51}} \\
4 & $0.86^\ddagger$ & $0.92^\ddagger$ & $0.96^\ddagger$ & $0.98^\ddagger$ & $0.99^\ddagger$ & $1.00^\ddagger$\\
5 & 0.021$^\ast$ & $0.14^\ast$ & $0.39^\ast$ & $0.62^\ast$ & $0.78^\ast$ & $0.89^\ast$\\
6 & --- & --- & --- & --- & $2^{-12.7\triangleleft}$ & $0.002^\triangleleft$\\
\bottomrule
\end{tabu}
\end{table}

\subsubsection{Usefulness of the preconditions.}

We now address the question of the usefulness of \cref{pcond41,pcond51} of \cref{sec:precond}.
Our goal is to determine with what probability randomly-generated
matrices in fact already satisfy the preconditions, and whether doing so for a matrix
${\bm\gamma}$ has a positive impact on its satisfying \cref{cond41} or \cref{cond51}.

We did this experimentally for two settings, both for the first scheme
corresponding to \cref{cond41}: order $d = 4$ over $\mathbb{F}_{2^8}$
and order $d = 5$ over $\mathbb{F}_{2^{13}}$. We generated enough random
matrices ${\bm\gamma}$ in order to obtain respectively 20\,000 and 2\,000 of
them satisfying \cref{cond41}, and counted how many of these corresponding safe
pairs (${\bm\gamma}$, $\om_{d\times d} - {\bm\gamma}$) had at least one or both elements that were MDS and XMDS. The same statistics were gathered for all the generated matrices, including the ones that were not safe.
The results are respectively summarized in \cref{tbl:stats_precond48,tbl:stats_precond513}.

\begin{table}[!htb]
\caption{\label{tbl:stats_precond48}Case $d = 4$ over
$\mathbb{F}_{2^8}$, for \Cref{cond41}.}
\begin{tabu} to \textwidth {XXXXXX}
\toprule
 & Total & One+ MDS & Both MDS & One+ XMDS & Both XMDS\\
\midrule
\#Random & 672\,625 & 634\,096 & 389\,504 & 515\,840 & 315\,273\\
\#Safe & 20\,000 & 19\,981 & 19\,981 & 19\,981 & 19\,981\\
Ratio & 0.030 & 0.032 & 0.051 & 0.039 & 0.063\\
\bottomrule
\end{tabu}
\end{table}

\begin{table}[!htb]
\caption{\label{tbl:stats_precond513}Case $d = 5$ over
$\mathbb{F}_{2^{13}}$, for \cref{cond41}.}
\begin{tabu} to \textwidth {XXXXXX}
\toprule
 & Total & One+ MDS & Both MDS & One+ XMDS & Both XMDS\\
\midrule
\#Random & 15\,877 & 15\,867 & 14\,978 & 15\,486 & 14\,623 \\
\#Safe & 2\,000 & 2\,000 & 2\,000 & 2\,000 & 2\,000\\
Ratio & 0.13 & 0.13 & 0.13 & 0.13 & 0.14\\
\bottomrule
\end{tabu}
\end{table}

A first comment on the results is that as already remarked in \cref{sec:precond},
the preconditions are not necessary to find safe instantiations. Indeed,
for a few of the smallest cases $d=3, q=2^3$ and $d=4,q=2^5$, we were
only able to find safe instantiations that did \emph{not} meet the
preconditions. For example, one can clearly see that the leading
$2\times 2$ submatrix of the following matrix is singular, and hence the
matrix is not MDS:
\[{\bm\gamma} = \begin{pmatrix}
\texttt{4} & \texttt{2} & \texttt{6}\\
\texttt{4} & \texttt{2} & \texttt{3}\\
\texttt{4} & \texttt{2} & \texttt{3}
\end{pmatrix}.\]
Yet (surprisingly), ${\bm\gamma}$ and $\om-{\bm\gamma}$ satisfy all
requirements of \cref{cond41} over $\mathbb{F}_{2^3}$.

Nonetheless, the precondition is clearly helpful in the vast majority of
cases. From our experiments, \emph{in cases where any preconditioned
safe matrix exists}, then nearly all safe matrices satisfy the
precondition,
while a significant fraction of random matrices do not. Enforcing the precondition by construction or as a first check is then indeed a way to improve the performance of a
random search of a safe matrix.
This is especially true for larger orders; for example, we did not find
any safe matrices for order $d=6$ over $\mathbb{F}_{2^{15}}$ by random
search, but only by imposing \cref{pcond41}.

Lastly, one should notice that specifically considering Cauchy matrices seems to further increase the odds of a matrix being safe, beyond the fact that it satisfies
\cref{cond41}: in the case $d = 4$, $\mathbb{F}_{2^8}$, \cref{tbl:stats_4} gives a success probability of 0.11, which is
significantly larger than the 0.063 of \cref{tbl:stats_precond48}, and in the case $d = 5$, $\mathbb{F}_{2^{13}}$, \cref{tbl:stats_4_2} gives 0.2, also quite higher than the 0.14 of \cref{tbl:stats_precond513}. As of yet, we do not have an explanation for this observation. 

\subsection{Instantiations of \cite[§4]{DBLP:conf/crypto/BelaidBPPTV17}}
\label{sec:inst41}

We conclude by giving explicit matrices allowing to safely instantiate the
scheme of~\cite[§4]{DBLP:conf/crypto/BelaidBPPTV17} over various fields
from order 3 up to 6; the case of order at most $2$ is treated
in \cref{sec:analytic} (Belaïd \etal also provided examples for $d=2$).
Our examples include practically-relevant instances with $d=3,4$ over $\mathbb{F}_{2^8}$.

We only give one matrix ${\bm\gamma}$ for every case we list, but we emphasise that as is required by the masking scheme, this means that both ${\bm\gamma}$ and ${\bm\delta} = \om_{d\times d} - {\bm\gamma}$ satisfy \cref{cond41}.
We list instances only for the smallest field size we know of, and for
$\mathbb{F}_{2^8}$ (when applicable), but have computed explicit
instances for all field sizes up to $\mathbb{F}_{2^{16}}$. The larger-field
instantiations are given in \cref{app:instantiations}.

\subsubsection{Instantiations at order 3.}

The smallest (binary) field for which we could find an instantiation at
order 3 was $\mathbb{F}_{2^3}$. Recall that we also have an explicit
construction in \cref{sec:analytic} for any $2^k$ with $k\ge 4$.

\[
{\bm\gamma}(\mathbb{F}_{2^3}) = \begin{pmatrix}
\texttt{3} & \texttt{5} & \texttt{4}\\
\texttt{3} & \texttt{6} & \texttt{7}\\
\texttt{3} & \texttt{5} & \texttt{4}
\end{pmatrix} \qquad
{\bm\gamma}(\mathbb{F}_{2^8}) = \begin{pmatrix}
\texttt{e3} & \texttt{b7} & \texttt{50}\\
\texttt{bd} & \texttt{e8} & \texttt{8b}\\
\texttt{53} & \texttt{25} & \texttt{a0}\\
\end{pmatrix}
\]

\subsubsection{Instantiations at order 4.}

The smallest (binary) field for which we could find an instantiation at order 4 was $\mathbb{F}_{2^5}$. The following matrices $\gam(\Fq)$ may be used
to instantiate the scheme over \Fq.

\[
\begin{array}{cc}
{\bm\gamma}(\mathbb{F}_{2^5}) = \begin{pmatrix}
\texttt{1c} & \texttt{ c} & \texttt{1e} & \texttt{ b}\\
\texttt{1c} & \texttt{ c} & \texttt{1e} & \texttt{12}\\
\texttt{10} & \texttt{18} & \texttt{17} & \texttt{14}\\
\texttt{1c} & \texttt{ c} & \texttt{1e} & \texttt{10}\\
\end{pmatrix} &
{\bm\gamma}(\mathbb{F}_{2^8}) = \begin{pmatrix}
\texttt{56} & \texttt{5e} & \texttt{a1} & \texttt{3d} \\
\texttt{97} & \texttt{27} & \texttt{71} & \texttt{c7} \\
\texttt{f5} & \texttt{ae} & \texttt{68} & \texttt{88} \\
\texttt{1c} & \texttt{ 3} & \texttt{9c} & \texttt{8e}
\end{pmatrix}\\
\end{array}
\]

\subsubsection{Instantiations at order 5.}
The smallest field for which we could find an instantiation at order 5 was
$\mathbb{F}_{2^{10}}$. The following matrix may be used
to instantiate the scheme over $\mathbb{F}_{2^{10}}$.

\[
\begin{array}{cc}
{\bm\gamma}(\mathbb{F}_{2^{10}}) = \begin{pmatrix}
\texttt{276} & \texttt{13e} & \texttt{ 64} & \texttt{1ab} & \texttt{120}\\
\texttt{189} & \texttt{181} & \texttt{195} & \texttt{30f} & \texttt{3fe}\\
\texttt{20a} & \texttt{3a1} & \texttt{199} & \texttt{ 30} & \texttt{2db}\\
\texttt{156} & \texttt{1ab} & \texttt{2f8} & \texttt{ e5} & \texttt{2a8}\\
\texttt{303} & \texttt{321} & \texttt{265} & \texttt{ d8} & \texttt{ 3a}\\
\end{pmatrix}\\
\end{array}
\]

\subsubsection{Instantiations at order 6.}
The smallest field for which we could find an instantiation at order 6 was
$\mathbb{F}_{2^{15}}$. The following matrix may be used
to instantiate the scheme over $\mathbb{F}_{2^{15}}$.

\[
{\bm\gamma}(\mathbb{F}_{2^{15}}) = \begin{pmatrix}
\texttt{151d} & \texttt{5895} & \texttt{5414} & \texttt{392b} & \texttt{2092} & \texttt{29a6}\\
\texttt{5c69} & \texttt{2f9e} & \texttt{241d} & \texttt{2ef7} & \texttt{ baa} & \texttt{6f40}\\
\texttt{6e0d} & \texttt{ 8cf} & \texttt{7ca1} & \texttt{6503} & \texttt{23dc} & \texttt{6b3b}\\
\texttt{10d7} & \texttt{588e} & \texttt{2c22} & \texttt{1245} & \texttt{6a38} & \texttt{6484}\\
\texttt{1637} & \texttt{7062} & \texttt{2ae0} & \texttt{ d1b} & \texttt{5305} & \texttt{381f}\\
\texttt{23f6} & \texttt{ 7d5} & \texttt{21bf} & \texttt{2879} & \texttt{2033} & \texttt{4377}\\
\end{pmatrix}
\]

\subsection{Instantiations of \cite[§5]{DBLP:conf/crypto/BelaidBPPTV17} up to order 6}
\label{sec:inst51}

We now give similar instantiation results for the scheme with linear randomness complexity. This time, only a single
matrix of dimension $(d+1)\times d$ is necessary to obtain a $d$-NI scheme. As in the previous case, we only focus here
on the cases where $3 \leq d \leq 6$, and only list the matrices over
the smallest field we have as well as $\mathbb{F}_{2^8}$ (where
possible). We refer to the supplementary material for all other cases.

\subsubsection{Instantiations at order 3.}

The smallest (binary) field for which we could find an instantiation at
order 3 was $\mathbb{F}_{2^3}$. Recall that we also have an explicit
construction in \cref{sec:analytic} for any $2^k$ with $k\ge 4$.

\[
{\bm\gamma}(\mathbb{F}_{2^3}) = \begin{pmatrix}
\texttt{1} & \texttt{7} & \texttt{4}\\
\texttt{4} & \texttt{4} & \texttt{4}\\
\texttt{2} & \texttt{1} & \texttt{4}\\
\texttt{7} & \texttt{2} & \texttt{4}\\
\end{pmatrix} \qquad
{\bm\gamma}(\mathbb{F}_{2^8}) = \begin{pmatrix}
\texttt{da} & \texttt{d5} & \texttt{e6}\\
\texttt{e8} & \texttt{1d} & \texttt{44}\\
\texttt{ad} & \texttt{b3} & \texttt{ce}\\
\texttt{9f} & \texttt{7b} & \texttt{6c}
\end{pmatrix}
\]

\subsubsection{Instantiations at order 4.}

The smallest (binary) field for which we could find an instantiation at order 4
was $\mathbb{F}_{2^5}$. The following matrices $\gam(\Fq)$ may be used
to instantiate the scheme over \Fq.
\[
\begin{array}{cc}
{\bm\gamma}(\mathbb{F}_{2^5}) = \begin{pmatrix}
\texttt{17} & \texttt{ f} & \texttt{13} & \texttt{16}\\
\texttt{ b} & \texttt{ 7} & \texttt{1a} & \texttt{11}\\
\texttt{ 1} & \texttt{1e} & \texttt{19} & \texttt{ 3}\\
\texttt{1b} & \texttt{10} & \texttt{ 2} & \texttt{ a}\\
\texttt{ 6} & \texttt{ 6} & \texttt{12} & \texttt{ e}
\end{pmatrix} &
{\bm\gamma}(\mathbb{F}_{2^8}) = \begin{pmatrix}
\texttt{ac} & \texttt{39} & \texttt{c0} & \texttt{36} \\
\texttt{79} & \texttt{5f} & \texttt{d9} & \texttt{51} \\
\texttt{9d} & \texttt{16} & \texttt{ca} & \texttt{63} \\
\texttt{a3} & \texttt{cb} & \texttt{ 6} & \texttt{81}\\
\texttt{eb} & \texttt{bb} & \texttt{d5} & \texttt{85}
\end{pmatrix}\\
\end{array}
\]

\subsubsection{Instantiations at order 5.}
The smallest field for which we could find an instantiation at order 5 was $\mathbb{F}_{2^{9}}$. The following matrix may be used
to instantiate the scheme over $\mathbb{F}_{2^{9}}$.

\[
\begin{array}{cc}
{\bm\gamma}(\mathbb{F}_{2^{9}}) = \begin{pmatrix}
\texttt{ 7d} & \texttt{12c} & \texttt{ 18} & \texttt{1a3} & \texttt{ da}\\
\texttt{121} & \texttt{131} & \texttt{109} & \texttt{1a7} & \texttt{ 3b}\\
\texttt{ 4a} & \texttt{131} & \texttt{ 91} & \texttt{ a4} & \texttt{1c4}\\
\texttt{17c} & \texttt{ cb} & \texttt{14b} & \texttt{ 41} & \texttt{ 57}\\
\texttt{ fd} & \texttt{ 87} & \texttt{ ac} & \texttt{17a} & \texttt{149}\\
\texttt{ 97} & \texttt{160} & \texttt{ 67} & \texttt{19b} & \texttt{ 3b}\\
\end{pmatrix}\\
\end{array}
\]

\subsubsection{Instantiations at order 6.}
The smallest field for which we could find an instantiation at order 6
was $\mathbb{F}_{2^{15}}$. The following matrix may be used
to instantiate the scheme over $\mathbb{F}_{2^{15}}$.

\[
{\bm\gamma}(\mathbb{F}_{2^{15}}) = \begin{pmatrix}
\texttt{475c} & \texttt{77e7} & \texttt{64ef} & \texttt{7893} & \texttt{4cd1} & \texttt{6e20}\\
\texttt{63dd} & \texttt{ 71f} & \texttt{29da} & \texttt{600e} & \texttt{36be} & \texttt{1db7}\\
\texttt{5511} & \texttt{ d63} & \texttt{3719} & \texttt{4874} & \texttt{ 664} & \texttt{5014}\\
\texttt{410e} & \texttt{7cf2} & \texttt{ 9d9} & \texttt{10a1} & \texttt{7525} & \texttt{6098}\\
\texttt{7bfe} & \texttt{2998} & \texttt{7e20} & \texttt{1438} & \texttt{35e6} & \texttt{ 51e}\\
\texttt{7564} & \texttt{75d3} & \texttt{221a} & \texttt{67c7} & \texttt{56f1} & \texttt{18d5}\\
\texttt{3e04} & \texttt{5d22} & \texttt{2fcf} & \texttt{33b7} & \texttt{6a39} & \texttt{5ed0}\\
\end{pmatrix}\\
\]

\subsection{Minimum field sizes for safe instantiations}

We conclude by briefly comparing the minimum field sizes for which we could find safe instantiations of \cref{cond41} and \cref{cond51} with the ones given by the non-constructive existence theorems
of Belaïd~\etal. Namely, \cite[Thm.~4.5]{DBLP:conf/crypto/BelaidBPPTV17} guarantees the existence of a pair of safe matrices for \cref{cond41} in dimension $d$ over \Fq as long as $q > 2d\cdot(12d)^d$,
and \cite[Thm.~5.4]{DBLP:conf/crypto/BelaidBPPTV17} of a safe matrix for \cref{cond51} as long as $q > d\cdot(d+1)\cdot(12d)^d$.
We give in \cref{tbl:insts} the explicit values provided by these two theorems for $2 \leq d \leq 6$ and $q$ a power of two, along with the experimental minima that we found. From these, it seems that
the sufficient condition of Belaïd~\etal is in fact rather pessimistic.

\begin{table}[!htb]
\caption{\label{tbl:insts}Sufficient field sizes for safe instantiations in characteristic two. Sizes are given as $\log(q)$.}
\begin{center}
\begin{tabu} to \textwidth {lcccc}
\toprule
$d\,/\,\min(\log(q))$\hspace{2em} & \cite[Thm.~4.5]{DBLP:conf/crypto/BelaidBPPTV17}~~~~& \cref{sec:inst41}~~~~~& \cite[Thm.~5.4]{DBLP:conf/crypto/BelaidBPPTV17}~~~~& \cref{sec:inst51} \\
\midrule
2 & 11 & 3  & 12 & 3\\
3 & 19 & 3  & 20 & 3\\
4 & 26 & 5  & 27 & 5\\
5 & 33 & 10 & 35 & 9\\
6 & 41 & 15 & 43 & 15\\
\bottomrule
\end{tabu}
\end{center}
\end{table}

\begin{center}
\decosix
\end{center}

\subsubsection*{Acknowledgements.}
We thank Daniel Augot for the interesting discussions we had in the early stages of this work.

This work was performed while the second author was graciously hosted by
the Laboratoire Jean Kuntzmann at the Universit\'e Grenoble Alpes.

The second author was supported in part by the
National Science Foundation under grants \#1319994 and \#1618269, and
by the Office of Naval Research award \#N0001417WX01516.

Some of the computations were performed using the Grace supercomputer
hosted by the U.S.\ Naval Academy Center for High Performance Computing,
with funding from the DoD HPC Modernization Program.

\bibliographystyle{amsalpha_nolower_eev2}

\newcommand{\etalchar}[1]{$^{#1}$}
\providecommand{\bysame}{\leavevmode\hbox to3em{\hrulefill}\thinspace}
\providecommand{\MR}{\relax\ifhmode\unskip\space\fi MR }
\providecommand{\MRhref}[2]{%
  \href{http://www.ams.org/mathscinet-getitem?mr=#1}{#2}
}
\providecommand{\href}[2]{#2}
\providecommand{\showEE}[1]{#1}
\providecommand{\titleEE}[2]{\ifthenelse{\equal{#2}{}}{\warning{missing
  title}}{}\ifthenelse{\equal{#1}{}}{#2}{\href{#1}{#2}}}

\appendix

\section{Explicit instantiations of schemes}
\label{app:instantiations}

We provide a complete listing of the safe ${\bm\gamma}$ matrices we have computed for both masking
schemes.

\subsection{Instantiations of \cite[§4]{DBLP:conf/crypto/BelaidBPPTV17}}

\subsubsection{Instantiations at order 3.}

The smallest (binary) field for which we could find an instantiation at
order 3 was $\mathbb{F}_{2^3}$. The following matrices $\gam(\Fq)$ may be used
to instantiate the scheme over \Fq.

\[
{\bm\gamma}(\mathbb{F}_{2^3}) = \begin{pmatrix}
\texttt{3} & \texttt{5} & \texttt{4}\\
\texttt{3} & \texttt{6} & \texttt{7}\\
\texttt{3} & \texttt{5} & \texttt{4}
\end{pmatrix} \quad
{\bm\gamma}(\mathbb{F}_{2^4}) = \begin{pmatrix}
\texttt{4} & \texttt{b} & \texttt{e}\\
\texttt{f} & \texttt{7} & \texttt{5}\\
\texttt{3} & \texttt{d} & \texttt{c}
\end{pmatrix} \quad
{\bm\gamma}(\mathbb{F}_{2^5}) = \begin{pmatrix}
\texttt{15} & \texttt{ 8} & \texttt{14}\\
\texttt{ f} & \texttt{1d} & \texttt{ c}\\
\texttt{16} & \texttt{ 7} & \texttt{ 5}
\end{pmatrix} \quad
{\bm\gamma}(\mathbb{F}_{2^6}) = \begin{pmatrix}
\texttt{36} & \texttt{30} & \texttt{1d}\\
\texttt{21} & \texttt{ 5} & \texttt{1a}\\
\texttt{35} & \texttt{31} & \texttt{1b}
\end{pmatrix}\]

\medskip

\[
{\bm\gamma}(\mathbb{F}_{2^7}) = \begin{pmatrix}
\texttt{7b} & \texttt{5a} & \texttt{11}\\
\texttt{64} & \texttt{5b} & \texttt{60}\\
\texttt{42} & \texttt{72} & \texttt{79}
\end{pmatrix} \qquad
{\bm\gamma}(\mathbb{F}_{2^8}) = \begin{pmatrix}
\texttt{e3} & \texttt{b7} & \texttt{50}\\
\texttt{bd} & \texttt{e8} & \texttt{8b}\\
\texttt{53} & \texttt{25} & \texttt{a0}
\end{pmatrix} \qquad
{\bm\gamma}(\mathbb{F}_{2^9}) = \begin{pmatrix}
\texttt{ c4} & \texttt{149} & \texttt{ 8c}\\
\texttt{112} & \texttt{167} & \texttt{ 5d}\\
\texttt{ da} & \texttt{110} & \texttt{13b}
\end{pmatrix}
\]

\medskip

\[
{\bm\gamma}(\mathbb{F}_{2^{10}}) = \begin{pmatrix}
\texttt{39f} & \texttt{2e4} & \texttt{2a9}\\
\texttt{ 67} & \texttt{25a} & \texttt{ 63}\\
\texttt{ 93} & \texttt{1d2} & \texttt{34a}
\end{pmatrix} \qquad
{\bm\gamma}(\mathbb{F}_{2^{11}}) = \begin{pmatrix}
\texttt{462} & \texttt{ 60} & \texttt{14b}\\
\texttt{3d5} & \texttt{3ce} & \texttt{1ab}\\
\texttt{ 22} & \texttt{223} & \texttt{11c}
\end{pmatrix} \qquad
{\bm\gamma}(\mathbb{F}_{2^{12}}) = \begin{pmatrix}
\texttt{7ef} & \texttt{ 7a} & \texttt{e06}\\
\texttt{3c9} & \texttt{be9} & \texttt{ca8}\\
\texttt{a7d} & \texttt{8b9} & \texttt{14d}
\end{pmatrix}
\]

\medskip

\[
{\bm\gamma}(\mathbb{F}_{2^{13}}) = \begin{pmatrix}
\texttt{720} & \texttt{ cff} & \texttt{1871}\\
\texttt{786} & \texttt{1596} & \texttt{ 37f}\\
\texttt{8bf} & \texttt{155e} & \texttt{ 8fc}
\end{pmatrix} \qquad
{\bm\gamma}(\mathbb{F}_{2^{14}}) = \begin{pmatrix}
\texttt{3c30} & \texttt{2f24} & \texttt{ 723}\\
\texttt{244b} & \texttt{3452} & \texttt{295c}\\
\texttt{1572} & \texttt{2682} & \texttt{1c92}
\end{pmatrix}
\]

\medskip

\[
{\bm\gamma}(\mathbb{F}_{2^{15}}) = \begin{pmatrix}
\texttt{4bf5} & \texttt{39c5} & \texttt{3929}\\
\texttt{  69} & \texttt{3f99} & \texttt{220e}\\
\texttt{40ad} & \texttt{7285} & \texttt{4538}
\end{pmatrix} \qquad
{\bm\gamma}(\mathbb{F}_{2^{16}}) = \begin{pmatrix}
\texttt{5ba1} & \texttt{264b} & \texttt{ 288}\\
\texttt{d51c} & \texttt{f2f7} & \texttt{43cb}\\
\texttt{22b0} & \texttt{ea98} & \texttt{4ddc}
\end{pmatrix}
\]

\subsubsection{Instantiations at order 4.}

The smallest (binary) field for which we could find an instantiation at order 4 was $\mathbb{F}_{2^5}$. The following matrices $\gam(\Fq)$ may be used
to instantiate the scheme over \Fq.

\[
\begin{array}{cc}
{\bm\gamma}(\mathbb{F}_{2^5}) = \begin{pmatrix}
\texttt{1c} & \texttt{ c} & \texttt{1e} & \texttt{ b}\\
\texttt{1c} & \texttt{ c} & \texttt{1e} & \texttt{12}\\
\texttt{10} & \texttt{18} & \texttt{17} & \texttt{14}\\
\texttt{1c} & \texttt{ c} & \texttt{1e} & \texttt{10}\\
\end{pmatrix} &
{\bm\gamma}(\mathbb{F}_{2^6}) = \begin{pmatrix}
\texttt{26} & \texttt{1b} & \texttt{ 8} & \texttt{3f}\\
\texttt{14} & \texttt{ 6} & \texttt{1e} & \texttt{2c}\\
\texttt{13} & \texttt{2a} & \texttt{33} & \texttt{22}\\
\texttt{3c} & \texttt{10} & \texttt{14} & \texttt{28}
\end{pmatrix}
\end{array}
\]

\medskip

\[\begin{array}{cc}
{\bm\gamma}(\mathbb{F}_{2^7}) = \begin{pmatrix}
\texttt{ e} & \texttt{6e} & \texttt{60} & \texttt{3d}\\
\texttt{51} & \texttt{27} & \texttt{6d} & \texttt{46}\\
\texttt{1d} & \texttt{21} & \texttt{43} & \texttt{13}\\
\texttt{48} & \texttt{2e} & \texttt{76} & \texttt{16}
\end{pmatrix}&
{\bm\gamma}(\mathbb{F}_{2^8}) = \begin{pmatrix}
\texttt{56} & \texttt{5e} & \texttt{a1} & \texttt{3d} \\
\texttt{97} & \texttt{27} & \texttt{71} & \texttt{c7} \\
\texttt{f5} & \texttt{ae} & \texttt{68} & \texttt{88} \\
\texttt{1c} & \texttt{ 3} & \texttt{9c} & \texttt{8e}
\end{pmatrix}\\
\end{array}
\]

\medskip

\[
\begin{array}{cc}
{\bm\gamma}(\mathbb{F}_{2^9}) = \begin{pmatrix}
\texttt{1b8} & \texttt{  30} & \texttt{1cf} & \texttt{  c}\\
\texttt{ fa} & \texttt{11d} &  \texttt{  f} & \texttt{16f}\\
\texttt{ 8f} & \texttt{ 56} &  \texttt{ 60} & \texttt{17f}\\
\texttt{104} & \texttt{ ec} &  \texttt{100} & \texttt{17e}\\
\end{pmatrix}&
{\bm\gamma}(\mathbb{F}_{2^{10}}) = \begin{pmatrix}
\texttt{23a} & \texttt{ ea} & \texttt{11b} & \texttt{16d}\\
\texttt{  9} & \texttt{3e2} & \texttt{387} & \texttt{197}\\
\texttt{2c4} & \texttt{148} & \texttt{296} & \texttt{1fc}\\
\texttt{14c} & \texttt{2c3} & \texttt{117} & \texttt{355}\\
\end{pmatrix}\\
\end{array}
\]

\medskip

\[
\begin{array}{cc}
{\bm\gamma}(\mathbb{F}_{2^{11}}) = \begin{pmatrix}
\texttt{36c} & \texttt{27a} & \texttt{32f} & \texttt{ 73}\\
\texttt{3bd} & \texttt{39d} & \texttt{610} & \texttt{254}\\
\texttt{3b1} & \texttt{27c} & \texttt{33a} & \texttt{3e4}\\
\texttt{42c} & \texttt{3f1} & \texttt{723} & \texttt{142}\\
\end{pmatrix}&
{\bm\gamma}(\mathbb{F}_{2^{12}}) = \begin{pmatrix}
\texttt{f19} & \texttt{ef4} & \texttt{16f} & \texttt{6b7}\\
\texttt{cfc} & \texttt{71c} & \texttt{b5d} & \texttt{f69}\\
\texttt{d23} & \texttt{440} & \texttt{b39} & \texttt{1e8}\\
\texttt{915} & \texttt{5c0} & \texttt{526} & \texttt{882}\\
\end{pmatrix}\\
\end{array}
\]

\medskip

\[
\begin{array}{cc}
{\bm\gamma}(\mathbb{F}_{2^{13}}) = \begin{pmatrix}
\texttt{ 4bf} & \texttt{ 559} & \texttt{ 1ef} & \texttt{ 2f2}\\
\texttt{ d75} & \texttt{1154} & \texttt{ fec} & \texttt{ a68}\\
\texttt{ a34} & \texttt{ ce6} & \texttt{ 41c} & \texttt{ e99}\\
\texttt{1941} & \texttt{18a0} & \texttt{ b83} & \texttt{17ae}\\
\end{pmatrix}&
{\bm\gamma}(\mathbb{F}_{2^{14}}) = \begin{pmatrix}
\texttt{ aa9} & \texttt{3b79} & \texttt{309e} & \texttt{258f}\\
\texttt{1711} & \texttt{1e67} & \texttt{1f6b} & \texttt{192b}\\
\texttt{ ecb} & \texttt{3c84} & \texttt{1cba} & \texttt{ da9}\\
\texttt{3b47} & \texttt{ 772} & \texttt{ 5cd} & \texttt{38c8}\\
\end{pmatrix}\\
\end{array}
\]

\medskip

\[
\begin{array}{cc}
{\bm\gamma}(\mathbb{F}_{2^{15}}) = \begin{pmatrix}
\texttt{2251} & \texttt{11d0} & \texttt{605a} & \texttt{63e6}\\
\texttt{7f22} & \texttt{68e6} & \texttt{ ed7} & \texttt{6bb7}\\
\texttt{487f} & \texttt{6fcf} & \texttt{5c3f} & \texttt{23ee}\\
\texttt{3b25} & \texttt{7289} & \texttt{19c4} & \texttt{50d4}\\
\end{pmatrix}&
{\bm\gamma}(\mathbb{F}_{2^{16}}) = \begin{pmatrix}
\texttt{4b5f} & \texttt{758b} & \texttt{ed70} & \texttt{40a2}\\
\texttt{9d32} & \texttt{ f21} & \texttt{6ca6} & \texttt{388e}\\
\texttt{8691} & \texttt{f39a} & \texttt{6def} & \texttt{860f}\\
\texttt{6576} & \texttt{897d} & \texttt{5020} & \texttt{b398}\\
\end{pmatrix}\\
\end{array}
\]

\subsubsection{Instantiations at order 5.}
The smallest field for which we could find an instantiation at order 5 was
$\mathbb{F}_{2^{10}}$. The following matrices $\gam(\Fq)$ may be used
to instantiate the scheme over \Fq.

\[
\begin{array}{cc}
{\bm\gamma}(\mathbb{F}_{2^{10}}) = \begin{pmatrix}
\texttt{276} & \texttt{13e} & \texttt{ 64} & \texttt{1ab} & \texttt{120}\\
\texttt{189} & \texttt{181} & \texttt{195} & \texttt{30f} & \texttt{3fe}\\
\texttt{20a} & \texttt{3a1} & \texttt{199} & \texttt{ 30} & \texttt{2db}\\
\texttt{156} & \texttt{1ab} & \texttt{2f8} & \texttt{ e5} & \texttt{2a8}\\
\texttt{303} & \texttt{321} & \texttt{265} & \texttt{ d8} & \texttt{ 3a}\\
\end{pmatrix}&
{\bm\gamma}(\mathbb{F}_{2^{11}}) = \begin{pmatrix}
\texttt{19d} & \texttt{57f} & \texttt{5b8} & \texttt{148} & \texttt{473}\\
\texttt{45f} & \texttt{176} & \texttt{517} & \texttt{1c9} & \texttt{2f7}\\
\texttt{171} & \texttt{699} & \texttt{41d} & \texttt{18e} & \texttt{5cb}\\
\texttt{6fe} & \texttt{ af} & \texttt{7a4} & \texttt{100} & \texttt{47d}\\
\texttt{482} & \texttt{181} & \texttt{441} & \texttt{44a} & \texttt{793}\\
\end{pmatrix}\\
\end{array}
\]

\medskip

\[
\begin{array}{cc}
{\bm\gamma}(\mathbb{F}_{2^{12}}) = \begin{pmatrix}
\texttt{866} & \texttt{440} & \texttt{a83} & \texttt{a02} & \texttt{b05}\\
\texttt{d77} & \texttt{449} & \texttt{a38} & \texttt{bd1} & \texttt{554}\\
\texttt{5b3} & \texttt{84a} & \texttt{a09} & \texttt{90c} & \texttt{c64}\\
\texttt{25e} & \texttt{c5f} & \texttt{d45} & \texttt{445} & \texttt{aa5}\\
\texttt{b56} & \texttt{5ac} & \texttt{4af} & \texttt{aa3} & \texttt{193}\\
\end{pmatrix}&
{\bm\gamma}(\mathbb{F}_{2^{13}}) = \begin{pmatrix}
\texttt{ 559} & \texttt{1ef} & \texttt{ 2f2} & \texttt{ 7c4} & \texttt{ 755}\\
\texttt{1154} & \texttt{fec} & \texttt{ a68} & \texttt{19f7} & \texttt{1c3b}\\
\texttt{ ce6} & \texttt{41c} & \texttt{ e99} & \texttt{10fc} & \texttt{1fda}\\
\texttt{18a0} & \texttt{b83} & \texttt{17ae} & \texttt{ 8bd} & \texttt{ f35}\\
\texttt{ c98} & \texttt{8fc} & \texttt{ efb} & \texttt{1200} & \texttt{14ae}\\
\end{pmatrix}\\
\end{array}
\]

\medskip

\[
\begin{array}{cc}
{\bm\gamma}(\mathbb{F}_{2^{14}}) = \begin{pmatrix}
\texttt{1ded} & \texttt{346c} & \texttt{2bc3} & \texttt{10d8} & \texttt{12be}\\
\texttt{2b47} & \texttt{3638} & \texttt{2032} & \texttt{3386} & \texttt{18f6}\\
\texttt{ 1a5} & \texttt{269a} & \texttt{ 70c} & \texttt{ 7e7} & \texttt{1c07}\\
\texttt{34bf} & \texttt{2462} & \texttt{ 8cf} & \texttt{1bd5} & \texttt{3941}\\
\texttt{3aef} & \texttt{3699} & \texttt{1faf} & \texttt{ cb2} & \texttt{3c41}\\
\end{pmatrix}&
{\bm\gamma}(\mathbb{F}_{2^{15}}) = \begin{pmatrix}
\texttt{3d33} & \texttt{3494} & \texttt{6bae} & \texttt{5d57} & \texttt{79e4}\\
\texttt{627a} & \texttt{ 1dd} & \texttt{ e95} & \texttt{3f5b} & \texttt{134c}\\
\texttt{ a03} & \texttt{4087} & \texttt{ b8c} & \texttt{31f0} & \texttt{75e8}\\
\texttt{4930} & \texttt{531b} & \texttt{4f33} & \texttt{2e8f} & \texttt{1a4c}\\
\texttt{1103} & \texttt{3dde} & \texttt{2834} & \texttt{1853} & \texttt{4754}\\
\end{pmatrix}\\
\end{array}
\]

\medskip

\[
{\bm\gamma}(\mathbb{F}_{2^{16}}) = \begin{pmatrix}
\texttt{758b} & \texttt{ed70} & \texttt{40a2} & \texttt{f1c7} & \texttt{9b8c}\\
\texttt{ f21} & \texttt{6ca6} & \texttt{388e} & \texttt{c9c9} & \texttt{1b09}\\
\texttt{f39a} & \texttt{6def} & \texttt{860f} & \texttt{d582} & \texttt{1cc3}\\
\texttt{897d} & \texttt{5020} & \texttt{b398} & \texttt{234b} & \texttt{2598}\\
\texttt{a9ea} & \texttt{f2ee} & \texttt{c8f3} & \texttt{1f04} & \texttt{ba18}\\
\end{pmatrix}
\]

\subsubsection{Instantiations at order 6.}
The smallest field for which we could find an instantiation at order 6 was $\mathbb{F}_{2^{15}}$. The following matrices $\gam(\Fq)$ may be used
to instantiate the scheme over \Fq.

\[
{\bm\gamma}(\mathbb{F}_{2^{15}}) = \begin{pmatrix}
\texttt{151d} & \texttt{5895} & \texttt{5414} & \texttt{392b} & \texttt{2092} & \texttt{29a6}\\
\texttt{5c69} & \texttt{2f9e} & \texttt{241d} & \texttt{2ef7} & \texttt{ baa} & \texttt{6f40}\\
\texttt{6e0d} & \texttt{ 8cf} & \texttt{7ca1} & \texttt{6503} & \texttt{23dc} & \texttt{6b3b}\\
\texttt{10d7} & \texttt{588e} & \texttt{2c22} & \texttt{1245} & \texttt{6a38} & \texttt{6484}\\
\texttt{1637} & \texttt{7062} & \texttt{2ae0} & \texttt{ d1b} & \texttt{5305} & \texttt{381f}\\
\texttt{23f6} & \texttt{ 7d5} & \texttt{21bf} & \texttt{2879} & \texttt{2033} & \texttt{4377}\\
\end{pmatrix}
\]
\[
{\bm\gamma}(\mathbb{F}_{2^{16}}) = \begin{pmatrix}
\texttt{9f80} & \texttt{97e3} & \texttt{1a0a} & \texttt{2dbf} & \texttt{93e7} & \texttt{c7a8}\\
\texttt{9dcf} & \texttt{3e14} & \texttt{ d5d} & \texttt{ec34} & \texttt{2375} & \texttt{28d6}\\
\texttt{4ee9} & \texttt{2f79} & \texttt{1bdd} & \texttt{1389} & \texttt{3f17} & \texttt{8803}\\
\texttt{1667} & \texttt{2d1f} & \texttt{d4ea} & \texttt{d573} & \texttt{49f6} & \texttt{697f}\\
\texttt{5877} & \texttt{2c2d} & \texttt{995d} & \texttt{a867} & \texttt{64e6} & \texttt{e758}\\
\texttt{e58c} & \texttt{c5a8} & \texttt{18cb} & \texttt{b3cd} & \texttt{a42b} & \texttt{722b}\\
\end{pmatrix}\\
\]

\subsection{Instantiations of \cite[§5]{DBLP:conf/crypto/BelaidBPPTV17}}

\subsubsection{Instantiations at order 3.}

The smallest (binary) field for which we could find an instantiation at
order 3 was $\mathbb{F}_{2^3}$. The following matrices $\gam(\Fq)$ may be used
to instantiate the scheme over \Fq.

\[
{\bm\gamma}(\mathbb{F}_{2^{3}}) = \begin{pmatrix}
\texttt{1} & \texttt{7} & \texttt{4}\\
\texttt{4} & \texttt{4} & \texttt{4}\\
\texttt{2} & \texttt{1} & \texttt{4}\\
\texttt{7} & \texttt{2} & \texttt{4}\\
\end{pmatrix} \quad
{\bm\gamma}(\mathbb{F}_{2^{4}}) = \begin{pmatrix}
\texttt{9} & \texttt{a} & \texttt{6}\\
\texttt{f} & \texttt{6} & \texttt{9}\\
\texttt{5} & \texttt{1} & \texttt{6}\\
\texttt{3} & \texttt{d} & \texttt{9}
\end{pmatrix} \quad
{\bm\gamma}(\mathbb{F}_{2^{5}}) = \begin{pmatrix}
\texttt{1b} & \texttt{ 9} & \texttt{ 4}\\
\texttt{ 5} & \texttt{13} & \texttt{1e}\\
\texttt{ e} & \texttt{1f} & \texttt{18}\\
\texttt{10} & \texttt{ 5} & \texttt{ 2}
\end{pmatrix} \quad
{\bm\gamma}(\mathbb{F}_{2^{6}}) = \begin{pmatrix}
\texttt{ c} & \texttt{25} & \texttt{3d}\\
\texttt{3f} & \texttt{2e} & \texttt{2c}\\
\texttt{24} & \texttt{ d} & \texttt{ 7}\\
\texttt{17} & \texttt{ 6} & \texttt{16}
\end{pmatrix}
\]

\medskip

\[
{\bm\gamma}(\mathbb{F}_{2^{7}}) = \begin{pmatrix}
\texttt{17} & \texttt{3c} & \texttt{1e}\\
\texttt{21} & \texttt{15} & \texttt{4e}\\
\texttt{35} & \texttt{14} & \texttt{16}\\
\texttt{ 3} & \texttt{3d} & \texttt{46}
\end{pmatrix} \qquad
{\bm\gamma}(\mathbb{F}_{2^{8}}) = \begin{pmatrix}
\texttt{da} & \texttt{d5} & \texttt{e6}\\
\texttt{e8} & \texttt{1d} & \texttt{44}\\
\texttt{ad} & \texttt{b3} & \texttt{ce}\\
\texttt{9f} & \texttt{7b} & \texttt{6c}
\end{pmatrix} \qquad
{\bm\gamma}(\mathbb{F}_{2^{9}}) = \begin{pmatrix}
\texttt{14b} & \texttt{ bd} & \texttt{ f6}\\
\texttt{ 62} & \texttt{ 4d} & \texttt{1b4}\\
\texttt{ 1a} & \texttt{124} & \texttt{18f}\\
\texttt{133} & \texttt{1d4} & \texttt{ cd}
\end{pmatrix}
\]

\medskip

\[
{\bm\gamma}(\mathbb{F}_{2^{10}}) = \begin{pmatrix}
\texttt{ 78} & \texttt{25b} & \texttt{ 97}\\
\texttt{35c} & \texttt{ ae} & \texttt{328}\\
\texttt{14c} & \texttt{292} & \texttt{ d2}\\
\texttt{268} & \texttt{ 67} & \texttt{36d}
\end{pmatrix} \qquad
{\bm\gamma}(\mathbb{F}_{2^{11}}) = \begin{pmatrix}
\texttt{111} & \texttt{1a5} & \texttt{50f}\\
\texttt{7c4} & \texttt{443} & \texttt{ 5a}\\
\texttt{697} & \texttt{76e} & \texttt{53b}\\
\texttt{ 42} & \texttt{288} & \texttt{ 6e}
\end{pmatrix} \qquad
{\bm\gamma}(\mathbb{F}_{2^{12}}) = \begin{pmatrix}
\texttt{91f} & \texttt{7b0} & \texttt{4c2}\\
\texttt{ad6} & \texttt{a47} & \texttt{7e3}\\
\texttt{743} & \texttt{3c4} & \texttt{ c8}\\
\texttt{48a} & \texttt{e33} & \texttt{3e9}
\end{pmatrix}
\]

\medskip

\[
{\bm\gamma}(\mathbb{F}_{2^{13}}) = \begin{pmatrix}
\texttt{1385} & \texttt{ fc8} & \texttt{153f}\\
\texttt{173d} & \texttt{1920} & \texttt{113a}\\
\texttt{ 40a} & \texttt{ 1b0} & \texttt{ 423}\\
\texttt{  b2} & \texttt{1758} & \texttt{  26}
\end{pmatrix} \qquad
{\bm\gamma}(\mathbb{F}_{2^{14}}) = \begin{pmatrix}
\texttt{3795} & \texttt{38e8} & \texttt{14fa}\\
\texttt{268a} & \texttt{ df7} & \texttt{27a2}\\
\texttt{ 259} & \texttt{359e} & \texttt{3cfe}\\
\texttt{1346} & \texttt{  81} & \texttt{ fa6}
\end{pmatrix}
\]

\medskip

\[
{\bm\gamma}(\mathbb{F}_{2^{15}}) = \begin{pmatrix}
\texttt{1852} & \texttt{ 689} & \texttt{305d}\\
\texttt{320d} & \texttt{33a4} & \texttt{3aaf}\\
\texttt{7873} & \texttt{4270} & \texttt{46d4}\\
\texttt{522c} & \texttt{775d} & \texttt{4c26}
\end{pmatrix} \qquad
{\bm\gamma}(\mathbb{F}_{2^{16}}) = \begin{pmatrix}
\texttt{4f70} & \texttt{6517} & \texttt{a398}\\
\texttt{e7a8} & \texttt{9d98} & \texttt{5b74}\\
\texttt{e251} & \texttt{3130} & \texttt{6ebf}\\
\texttt{4a89} & \texttt{c9bf} & \texttt{9653}
\end{pmatrix}
\]

\subsubsection{Instantiations at order 4.}

The smallest (binary) field for which we could find an instantiation at order 4 was $\mathbb{F}_{2^5}$. The following matrices $\gam(\Fq)$ may be used
to instantiate the scheme over \Fq.
\[
\begin{array}{cc}
{\bm\gamma}(\mathbb{F}_{2^5}) = \begin{pmatrix}
\texttt{17} & \texttt{ f} & \texttt{13} & \texttt{16}\\
\texttt{ b} & \texttt{ 7} & \texttt{1a} & \texttt{11}\\
\texttt{ 1} & \texttt{1e} & \texttt{19} & \texttt{ 3}\\
\texttt{1b} & \texttt{10} & \texttt{ 2} & \texttt{ a}\\
\texttt{ 6} & \texttt{ 6} & \texttt{12} & \texttt{ e}
\end{pmatrix} &
{\bm\gamma}(\mathbb{F}_{2^6}) = \begin{pmatrix}
\texttt{ f} & \texttt{2f} & \texttt{20} & \texttt{25}\\
\texttt{1c} & \texttt{28} & \texttt{ 6} & \texttt{25}\\
\texttt{32} & \texttt{2c} & \texttt{ 9} & \texttt{ 8}\\
\texttt{26} & \texttt{28} & \texttt{11} & \texttt{13}\\
\texttt{ 7} & \texttt{ 3} & \texttt{3e} & \texttt{1b}
\end{pmatrix}\\
\end{array}
\]

\medskip

\[
\begin{array}{cc}
{\bm\gamma}(\mathbb{F}_{2^7}) = \begin{pmatrix}
\texttt{7f} & \texttt{14} & \texttt{50} & \texttt{5f}\\
\texttt{35} & \texttt{58} & \texttt{45} & \texttt{6b}\\
\texttt{24} & \texttt{60} & \texttt{5e} & \texttt{2e}\\
\texttt{11} & \texttt{1e} & \texttt{2d} & \texttt{7b}\\
\texttt{7f} & \texttt{32} & \texttt{66} & \texttt{61}
\end{pmatrix}&
{\bm\gamma}(\mathbb{F}_{2^8}) = \begin{pmatrix}
\texttt{ac} & \texttt{39} & \texttt{c0} & \texttt{36} \\
\texttt{79} & \texttt{5f} & \texttt{d9} & \texttt{51} \\
\texttt{9d} & \texttt{16} & \texttt{ca} & \texttt{63} \\
\texttt{a3} & \texttt{cb} & \texttt{ 6} & \texttt{81}\\
\texttt{eb} & \texttt{bb} & \texttt{d5} & \texttt{85}
\end{pmatrix}\\
\end{array}
\]

\medskip

\[
\begin{array}{cc}
{\bm\gamma}(\mathbb{F}_{2^9}) = \begin{pmatrix}
\texttt{3e} & \texttt{1e0} & \texttt{  5} & \texttt{1ef}\\
\texttt{ e} & \texttt{ 19} &  \texttt{180} & \texttt{ c4}\\
\texttt{93} & \texttt{186} &  \texttt{ d9} & \texttt{ 98}\\
\texttt{82} & \texttt{ 49} &  \texttt{ 36} & \texttt{191}\\
\texttt{21} & \texttt{ 36} &  \texttt{16a} & \texttt{ 22}\\
\end{pmatrix}&
{\bm\gamma}(\mathbb{F}_{2^{10}}) = \begin{pmatrix}
\texttt{ad} & \texttt{244} & \texttt{388} & \texttt{1d3}\\
\texttt{7a} & \texttt{253} & \texttt{ 32} & \texttt{3d4}\\
\texttt{b2} & \texttt{370} & \texttt{128} & \texttt{1cc}\\
\texttt{41} & \texttt{ b7} & \texttt{2c0} & \texttt{390}\\
\texttt{24} & \texttt{3d0} & \texttt{ 52} & \texttt{ 5b}\\
\end{pmatrix}\\
\end{array}
\]

\medskip

\[
\begin{array}{cc}
{\bm\gamma}(\mathbb{F}_{2^{11}}) = \begin{pmatrix}
\texttt{6a7} & \texttt{ e6} & \texttt{ ee} & \texttt{ 5c}\\
\texttt{13d} & \texttt{29e} & \texttt{781} & \texttt{7cd}\\
\texttt{225} & \texttt{75a} & \texttt{534} & \texttt{25b}\\
\texttt{25a} & \texttt{364} & \texttt{479} & \texttt{37d}\\
\texttt{7e5} & \texttt{646} & \texttt{622} & \texttt{6b7}\\
\end{pmatrix}&
{\bm\gamma}(\mathbb{F}_{2^{12}}) = \begin{pmatrix}
\texttt{4db} & \texttt{48a} & \texttt{5b9} & \texttt{83e}\\
\texttt{f2e} & \texttt{616} & \texttt{941} & \texttt{725}\\
\texttt{58a} & \texttt{b17} & \texttt{543} & \texttt{ 3e}\\
\texttt{ 6c} & \texttt{243} & \texttt{caf} & \texttt{aab}\\
\texttt{e13} & \texttt{bc8} & \texttt{514} & \texttt{58e}\\
\end{pmatrix}\\
\end{array}
\]

\medskip

\[
\begin{array}{cc}
{\bm\gamma}(\mathbb{F}_{2^{13}}) = \begin{pmatrix}
\texttt{ fa9} & \texttt{ 50f} & \texttt{1f87} & \texttt{ a97}\\
\texttt{181e} & \texttt{ 1cf} & \texttt{1725} & \texttt{ 86c}\\
\texttt{ e22} & \texttt{ 8eb} & \texttt{1800} & \texttt{118d}\\
\texttt{168f} & \texttt{ e76} & \texttt{1f81} & \texttt{ e8d}\\
\texttt{ f1a} & \texttt{ 25d} & \texttt{ f23} & \texttt{1dfb}\\
\end{pmatrix}&
{\bm\gamma}(\mathbb{F}_{2^{14}}) = \begin{pmatrix}
\texttt{261d} & \texttt{  ff} & \texttt{1fcb} & \texttt{ ae1}\\
\texttt{ 4f8} & \texttt{3575} & \texttt{1be2} & \texttt{ ea6}\\
\texttt{139a} & \texttt{3353} & \texttt{3ca8} & \texttt{116c}\\
\texttt{2d98} & \texttt{1eb9} & \texttt{ 9d7} & \texttt{3fad}\\
\texttt{1ce7} & \texttt{1860} & \texttt{3156} & \texttt{2a86}\\
\end{pmatrix}\\
\end{array}
\]

\medskip

\[
\begin{array}{cc}
{\bm\gamma}(\mathbb{F}_{2^{15}}) = \begin{pmatrix}
\texttt{246d} & \texttt{79de} & \texttt{632b} & \texttt{2b2f}\\
\texttt{1fe9} & \texttt{3986} & \texttt{13da} & \texttt{6a77}\\
\texttt{4e15} & \texttt{6f28} & \texttt{4e9a} & \texttt{2778}\\
\texttt{5389} & \texttt{6a45} & \texttt{7849} & \texttt{7770}\\
\texttt{2618} & \texttt{4535} & \texttt{4622} & \texttt{1150}\\
\end{pmatrix}&
{\bm\gamma}(\mathbb{F}_{2^{16}}) = \begin{pmatrix}
\texttt{dfd3} & \texttt{a0b4} & \texttt{ca3b} & \texttt{39bb}\\
\texttt{b92f} & \texttt{f0a7} & \texttt{b829} & \texttt{bf8d}\\
\texttt{ae71} & \texttt{3990} & \texttt{7757} & \texttt{3943}\\
\texttt{5bd5} & \texttt{f925} & \texttt{ 188} & \texttt{af4f}\\
\texttt{9358} & \texttt{90a6} & \texttt{ 4cd} & \texttt{103a}\\
\end{pmatrix}\\
\end{array}
\]

\subsubsection{Instantiations at order 5.}
The smallest field for which we could find an instantiation at order 5 was $\mathbb{F}_{2^{9}}$. The following matrices $\gam(\Fq)$ may be used
to instantiate the scheme over \Fq.

\[
\begin{array}{cc}
{\bm\gamma}(\mathbb{F}_{2^{9}}) = \begin{pmatrix}
\texttt{ 7d} & \texttt{12c} & \texttt{ 18} & \texttt{1a3} & \texttt{ da}\\
\texttt{121} & \texttt{131} & \texttt{109} & \texttt{1a7} & \texttt{ 3b}\\
\texttt{ 4a} & \texttt{131} & \texttt{ 91} & \texttt{ a4} & \texttt{1c4}\\
\texttt{17c} & \texttt{ cb} & \texttt{14b} & \texttt{ 41} & \texttt{ 57}\\
\texttt{ fd} & \texttt{ 87} & \texttt{ ac} & \texttt{17a} & \texttt{149}\\
\texttt{ 97} & \texttt{160} & \texttt{ 67} & \texttt{19b} & \texttt{ 3b}\\
\end{pmatrix}&
{\bm\gamma}(\mathbb{F}_{2^{10}}) = \begin{pmatrix}
\texttt{ 33} & \texttt{314} & \texttt{2b6} & \texttt{ 4d} & \texttt{236}\\
\texttt{285} & \texttt{339} & \texttt{ 8a} & \texttt{3bb} & \texttt{ 79}\\
\texttt{ 56} & \texttt{118} & \texttt{ b6} & \texttt{373} & \texttt{326}\\
\texttt{132} & \texttt{1b5} & \texttt{2cd} & \texttt{  7} & \texttt{335}\\
\texttt{ 72} & \texttt{ d4} & \texttt{101} & \texttt{26e} & \texttt{10e}\\
\texttt{3a0} & \texttt{ 54} & \texttt{146} & \texttt{2ec} & \texttt{352}\\
\end{pmatrix}\\
\end{array}
\]

\medskip

\[
\begin{array}{cc}
{\bm\gamma}(\mathbb{F}_{2^{11}}) = \begin{pmatrix}
\texttt{1ce} & \texttt{ d9} & \texttt{5d5} & \texttt{690} & \texttt{6ae}\\
\texttt{176} & \texttt{7fa} & \texttt{44e} & \texttt{559} & \texttt{ a2}\\
\texttt{3e2} & \texttt{532} & \texttt{ c9} & \texttt{ 7a} & \texttt{447}\\
\texttt{4f1} & \texttt{ 4d} & \texttt{64b} & \texttt{ 36} & \texttt{ 65}\\
\texttt{ bc} & \texttt{26d} & \texttt{1cc} & \texttt{645} & \texttt{ 84}\\
\texttt{717} & \texttt{ 31} & \texttt{6d5} & \texttt{5c0} & \texttt{2aa}\\
\end{pmatrix}&
{\bm\gamma}(\mathbb{F}_{2^{12}}) = \begin{pmatrix}
\texttt{8ef} & \texttt{276} & \texttt{61a} & \texttt{b58} & \texttt{2ab}\\
\texttt{d02} & \texttt{ 63} & \texttt{871} & \texttt{ 61} & \texttt{cb8}\\
\texttt{4da} & \texttt{ d8} & \texttt{ced} & \texttt{3f5} & \texttt{ce6}\\
\texttt{bc3} & \texttt{d44} & \texttt{c82} & \texttt{ 1a} & \texttt{c2a}\\
\texttt{ c6} & \texttt{125} & \texttt{ed3} & \texttt{9fc} & \texttt{906}\\
\texttt{a32} & \texttt{eac} & \texttt{ d7} & \texttt{12a} & \texttt{7d9}\\
\end{pmatrix}\\
\end{array}
\]

\medskip

\[
\begin{array}{cc}
{\bm\gamma}(\mathbb{F}_{2^{13}}) = \begin{pmatrix}
\texttt{ 89a} & \texttt{1c76} & \texttt{ e56} & \texttt{ ae5} & \texttt{ a19}\\
\texttt{14c4} & \texttt{ 20c} & \texttt{ 198} & \texttt{13f1} & \texttt{ 886}\\
\texttt{ 6bf} & \texttt{ e58} & \texttt{1ed8} & \texttt{1ae3} & \texttt{19fb}\\
\texttt{ 519} & \texttt{1171} & \texttt{1c43} & \texttt{10e7} & \texttt{ f50}\\
\texttt{ fd5} & \texttt{13de} & \texttt{ c24} & \texttt{1f01} & \texttt{1a9d}\\
\texttt{102d} & \texttt{128d} & \texttt{ 171} & \texttt{ c11} & \texttt{ ea9}\\
\end{pmatrix}&
{\bm\gamma}(\mathbb{F}_{2^{14}}) = \begin{pmatrix}
\texttt{1d03} & \texttt{3719} & \texttt{39b0} & \texttt{3a21} & \texttt{3598}\\
\texttt{ 550} & \texttt{ 82a} & \texttt{3f3f} & \texttt{2aba} & \texttt{35cb}\\
\texttt{3f2f} & \texttt{3a81} & \texttt{1109} & \texttt{37f0} & \texttt{2175}\\
\texttt{23c2} & \texttt{194a} & \texttt{ dc6} & \texttt{3fa3} & \texttt{29a4}\\
\texttt{3e3f} & \texttt{ 571} & \texttt{23c6} & \texttt{31ee} & \texttt{3c23}\\
\texttt{3a81} & \texttt{1989} & \texttt{3986} & \texttt{2926} & \texttt{34a1}\\
\end{pmatrix}\\
\end{array}
\]

\medskip

\[
\begin{array}{cc}
{\bm\gamma}(\mathbb{F}_{2^{15}}) = \begin{pmatrix}
\texttt{ 2bd} & \texttt{662d} & \texttt{3f88} & \texttt{5519} & \texttt{6e67}\\
\texttt{4519} & \texttt{71cc} & \texttt{44a5} & \texttt{102c} & \texttt{3f61}\\
\texttt{313c} & \texttt{160f} & \texttt{131b} & \texttt{6695} & \texttt{4631}\\
\texttt{2c83} & \texttt{53b7} & \texttt{1b64} & \texttt{504b} & \texttt{ dd1}\\
\texttt{4733} & \texttt{1baa} & \texttt{11a4} & \texttt{ b15} & \texttt{46ff}\\
\texttt{1d28} & \texttt{49f3} & \texttt{62f6} & \texttt{78fe} & \texttt{5c19}\\
\end{pmatrix}
{\bm\gamma}(\mathbb{F}_{2^{16}}) = \begin{pmatrix}
\texttt{f4ff} & \texttt{3efb} & \texttt{b917} & \texttt{5dab} & \texttt{c491}\\
\texttt{9179} & \texttt{d251} & \texttt{abbd} & \texttt{544d} & \texttt{426b}\\
\texttt{3242} & \texttt{e774} & \texttt{cc82} & \texttt{2de0} & \texttt{  55}\\
\texttt{ d5e} & \texttt{2439} & \texttt{28ca} & \texttt{539f} & \texttt{c5ab}\\
\texttt{9659} & \texttt{1cbc} & \texttt{7431} & \texttt{2eae} & \texttt{f356}\\
\texttt{ccc3} & \texttt{335b} & \texttt{82d3} & \texttt{5937} & \texttt{b052}\\
\end{pmatrix}\\
\end{array}
\]

\subsubsection{Instantiations at order 6.}
The smallest field for which we could find an instantiation at order 6
was $\mathbb{F}_{2^{15}}$. The following matrices may be used
to instantiate the scheme over $\Fq$.

\[
{\bm\gamma}(\mathbb{F}_{2^{15}}) = \begin{pmatrix}
\texttt{475c} & \texttt{77e7} & \texttt{64ef} & \texttt{7893} & \texttt{4cd1} & \texttt{6e20}\\
\texttt{63dd} & \texttt{ 71f} & \texttt{29da} & \texttt{600e} & \texttt{36be} & \texttt{1db7}\\
\texttt{5511} & \texttt{ d63} & \texttt{3719} & \texttt{4874} & \texttt{ 664} & \texttt{5014}\\
\texttt{410e} & \texttt{7cf2} & \texttt{ 9d9} & \texttt{10a1} & \texttt{7525} & \texttt{6098}\\
\texttt{7bfe} & \texttt{2998} & \texttt{7e20} & \texttt{1438} & \texttt{35e6} & \texttt{ 51e}\\
\texttt{7564} & \texttt{75d3} & \texttt{221a} & \texttt{67c7} & \texttt{56f1} & \texttt{18d5}\\
\texttt{3e04} & \texttt{5d22} & \texttt{2fcf} & \texttt{33b7} & \texttt{6a39} & \texttt{5ed0}\\
\end{pmatrix}\\
\]

\medskip

\[
{\bm\gamma}(\mathbb{F}_{2^{16}}) = \begin{pmatrix}
\texttt{d997} & \texttt{8a77} & \texttt{f6eb} & \texttt{b902} & \texttt{a02d} & \texttt{f8f6}\\
\texttt{a7b9} & \texttt{239c} & \texttt{c977} & \texttt{8270} & \texttt{7b14} & \texttt{34a8}\\
\texttt{571c} & \texttt{bc5c} & \texttt{539b} & \texttt{c981} & \texttt{16a4} & \texttt{ff58}\\
\texttt{9417} & \texttt{b095} & \texttt{f080} & \texttt{e399} & \texttt{d925} & \texttt{687b}\\
\texttt{5f28} & \texttt{6048} & \texttt{cf5a} & \texttt{1158} & \texttt{2db9} & \texttt{b4e1}\\
\texttt{8ae1} & \texttt{75e7} & \texttt{fb1c} & \texttt{77e9} & \texttt{22ec} & \texttt{74fb}\\
\texttt{68ec} & \texttt{b08d} & \texttt{a8c1} & \texttt{77db} & \texttt{1bed} & \texttt{9b67}\\
\end{pmatrix}\\
\]

\end{document}